\documentclass[a4paper,UKenglish,numberwithinsect,cleveref]{lipics-v2019}
\pdfoutput=1

 \usepackage[utf8]{inputenc}
\usepackage{mathtools} 
\usepackage{stmaryrd} 
\usepackage{amsmath}
\usepackage{amsthm}
\usepackage{amssymb}
\usepackage{tikz}
\usepackage{tikz-cd}
\usepackage{enumitem}
\usepackage{centernot}
\usepackage{url}

\newtheorem{prop}[theorem]{Proposition}
\newtheorem{cor}[theorem]{Corollary}

\newlist{romani}{enumerate}{1}
\setlist[romani]{label={\roman*)}}

\newcommand{\cutout}[1]{}

\newcommand{\Iff}{\Longleftrightarrow}
\newcommand{\ra}{\rightarrow}

\newcommand{\E}{\exists}
\newcommand{\ex}{\exists}

\newcommand{\all}{\forall}

\renewcommand{\phi}{\varphi}
\renewcommand{\theta}{\vartheta}
\renewcommand{\emptyset}{\varnothing}
\renewcommand{\epsilon}{\varepsilon}

\newcommand{\free}{{\rm free}}
\newcommand{\FO}{{\rm FO}}

\newcommand{\WeakDet}{\mathrm{WeakDet}}
\newcommand{\StrongDet}{\mathrm{StrongDet}}
\newcommand{\NoSig}{\mathrm{NoSig}}
\newcommand{\SingVal}{\mathrm{SingVal}}
\newcommand{\LInd}{\lambda\text{\rm -Indep}}
\newcommand{\OutInd}{\text{\rm Out-Indep}}
\newcommand{\ParInd}{\text{\rm Par-Indep}}
\newcommand{\Loc}{\mathrm{Loc}}
\newcommand{\NonContext}{\mathrm{NonContext}}

\newcommand*{\tup}[1]{\bar{#1}}

\newcommand{\tx}{\tup x}
\newcommand{\ty}{\tup y}
\newcommand{\tz}{\tup z}

\def\reals{\mathbb{R}}

\def\nats{\mathbb{N}}

\def\XX{\mathbb{X}}
\def\YY{\mathbb{Y}}
\def\ZZ{\mathbb{Z}}
\def\Pr{\mathbb{P}}

\newcommand{\dom}{\mathrm{dom}}

\newcommand{\bland}{\bigwedge}
\newcommand{\blor}{\bigvee}

\DeclareMathOperator{\Var}{Var} 
\newcommand{\VarH}{\Var_n^h} 
\newcommand{\VarE}{\Var_n^e} 

\DeclareMathOperator{\dep}{dep} 
\newcommand{\pr}[1]{\left(#1\right)} 

\newcommand{\calH}{\mathcal{H}} 

\newcommand{\Ff}{{\cal F}}
\newcommand{\Pp}{{\cal P}}

\newcommand{\frA}{\mathfrak{A}}

\newcommand{\ph}{\varphi}

\newcommand{\sub}{\subseteq}

\newcommand{\set}[1]{\{ #1 \}}

\DeclarePairedDelimiter\abs{\lvert}{\rvert}

\newcommand{\relEnt}{\models_{\text{rel}}}
\newcommand{\probEnt}{\models_{\text{prob}}}
\newcommand{\allEnt}{\models_{\text{all}}}

\def\impl{\implies}

\newcommand{\nimpl}{\centernot\impl}

\makeatletter
\providecommand{\doi}[1]{%
  \begingroup
    \let\bibinfo\@secondoftwo
    \urlstyle{rm}%
    \href{http://dx.doi.org/#1}{%
      doi:\discretionary{}{}{}%
      \nolinkurl{#1}%
    }%
  \endgroup
}
\makeatother

\crefname{appendix}{}{}

\title{Unifying Hidden-Variable Problems from Quantum Mechanics by Logics of Dependence and Independence}
\titlerunning{Hidden-Variable Problems and Logics of Dependence and Independence}

\author{Rafael Albert}{RWTH Aachen University, Germany}{rafael.albert@rwth-aachen.de}{}{}

\author{Erich Grädel}{RWTH Aachen University, Germany}{graedel@logic.rwth-aachen.de}{}{}

\authorrunning{R. Albert and E. Grädel}

\Copyright{R. Albert, E. Grädel}

\ccsdesc[500]{Theory of computation~Logic}

\keywords{Hidden-variables, logics of dependence and independence, relational
  versus probabilistic team semantics, Kochen-Specker Theorem, Bell's Theorem}

\nolinenumbers 
\hideLIPIcs  

\begin{document}

\maketitle

\begin{abstract}
We study hidden-variable models from quantum mechanics and their abstractions in purely probabilistic and 
relational frameworks by means of logics of dependence and independence, which
are based on team semantics.
We show that common desirable properties of hidden-variable models can be defined in an elegant
and concise way in dependence and independence logic. The relationship between different properties
and their simultaneous realisability can thus be formulated and proven on a purely logical level,
as problems of entailment and satisfiability of logical formulae.
Connections between probabilistic and relational
entailment in dependence and independence logic allow us to
simplify proofs. In many cases, we can establish results on both probabilistic and relational
hidden-variable models by a single proof, because one case implies the other,
depending on purely syntactic criteria.
We also discuss the  `no-go' theorems by Bell and Kochen-Specker and
provide a purely logical variant of the latter, introducing 
\emph{non-contextual choice} as a team-semantical property.  
\end{abstract}

\section{Introduction}

Hidden-variable models have been proposed since the 1920s as an alternative to the
standard interpretation, sometimes called the Copenhagen interpretation, of quantum
mechanics with the goal to explain and remove counterintuitive aspects of
quantum mechanics. In particular, quantum systems behave, according to the
standard interpretation, probabilistically rather than deterministically,
non-local interactions between agents or particles that are widely separated in
space are possible through entanglement, and there is an unavoidable dependence
between an observer of a quantum mechanical system and the observed properties.
Due to such features, Einstein, Podolsky, and Rosen \cite{EinsteinPodRos35}
considered a quantum mechanical state an `incomplete'
description of physical reality.
The basic idea of hidden-variable models is to `complete' quantum mechanics by adding unobservable,
`hidden', variables to the description of a system, to obtain models that are
consistent with the predictions of quantum mechanics, but which do not exhibit
counterintuitive behavior.

\medskip
The question to what extent hidden-variable models can indeed
explain quantum mechanical effects in a satisfactory way, has been studied for
a long time by many researchers. John von Neumann \cite{Neumann32}, who 
has coined the term `hidden-variable models', claimed to have 
established the impossibility of such an endeavour:
\begin{quote}
It should be noted that we need not go any further into the mechanism of the ‘hidden parameters,’ since we now know that the established results of quantum mechanics can never be re-derived with their help.
\end{quote}
However, this claim was not generally accepted. Some of the most outspoken criticism came from
Grete Hermann, Davin Mermin, and John Bell \cite{Bell66}, who
even went as far as calling von Neumann's proof  ``not merely false, but foolish''
(see \cite{Bub10}). Further, David Bohm and Lois de Broglie developed 
non-standard,  deterministic interpretations of quantum mechanics using hidden-variables
referred to as Bohmian mechanics or de Broglie-Bohm theory
\cite{Bohm52a,Bohm52b,deBroglie27}.

\medskip
Many different variants of determinism, locality, and independence properties of hidden-variable models have been studied. 
The attempts to build hidden-variable models that jointly realise these 
can be seen as attempts  to avoid the counter-intuitive consequences of the standard 
formalism of quantum
mechanics. However, the famous `no-go'  theorems of quantum mechanics, such as the ones by 
Bell \cite{Bell64} and Kochen-Specker \cite{KochenSpe67}, show that one cannot really escape
such consequences, and that there are severe limitations for the hidden-variable programme in general.
The dominant attitude today is to accept the quantum mechanical formalism at face value and
to make use of its counter-intuitive features such as entanglement and contextuality in
modern applications like quantum computing and quantum cryptology. 

A survey of the most important properties of determinism, locality, and independence in hidden-variable models 
has been given in an influential paper by  Brandenburger and Yanofsky  \cite{BrandenburgerYan08}, in 
an abstract and purely probabilistic framework that does not make explicit reference to quantum mechanics. 
It  makes precise the relationship between the different properties and discusses the question
which combinations of them can be realised simultaneously in a probabilistic hidden-variable model.
Despite the general limitations of the hidden-variable programme,
there are interesting combinations of properties that are jointly realisable, 
and the work  of Brandenburger and Yanovsky \cite{BrandenburgerYan08}
gives a detailed account of what is possible in a probabilistic setting. 
A further fundamental step for the understanding of hidden-variable phenomena has been Abramsky's
proposal of a purely relational (rather than probabilistic) framework \cite{Abramsky13}
for hidden-variable models, with discrete analogues of the probabilistic dependence
and independence properties studied in \cite{BrandenburgerYan08}. He showed that the main 
structure of the theory is preserved under this simplification. Abramsky's work 
opens the possibility to study hidden-variable questions on a much more general level
leaving aside quantum mechanical details. But on the other side, Abramsky also discusses the 
issue of quantum realisations of
empirical models (which has also been studied on a more general level in \cite{AbramskyBra11}).
An important point is that among the  emipirical models that cannot be extended to hidden-variable models
with specific locality and independence properties there are models that are not just abstract
mathematical constructions but do indeed arise from quantum mechanical experiments, such as Bell test. 

\medskip We propose here a study of hidden-variable properties by means of modern logics
of dependence and independence.
These logics are based on team semantics, introduced by Hodges \cite{Hodges97}.
While a classical logical formula (from first-order logic, for instance) is evaluated
for a single assignment, mapping its free variables  to values in some
mathematical structure, team semantics evaluates a formula for a \emph{set} of such assignments,
called a team. Further, while previous formalisms for studying dependence and independence
(such as Henkin quantifiers or independence-friendly logic) had modeled dependencies by
special quantifiers, the modern dependence and independence  
logics treat them, following a proposal by Väänänen \cite{Vaananen07},  
as \emph{atomic properties of teams}. While these logics are syntactically 
very simple, extending the atomic team properties by the standard first-order operators
$\lor$, $\land$, $\exists$ and $\forall$, 
they are semantically rather powerful. Indeed, team semantics admits 
the manipulation of second-order objects by first-order syntax and, 
in fact, independence logic \cite{GraedelVaa13} has the full expressive power
of existential second-order and can thus define all team 
properties in the complexity class NP \cite{Galliani12}.

\medskip
There are several reasons why logics of dependence and independence are 
natural tools for reasoning about hidden-variable models.
First of all, the desirable properties of hidden-variable models, in particular the ones studied
in \cite{Abramsky13} and \cite{BrandenburgerYan08}, are obviously properties of dependence
and independence. We shall see that in the logics we are using, their definition is extremely simple and transparent,
mostly just conjunctions of dependence or independence atoms. Second, 
the models studied in hidden-variable theories, both empirical models and 
hidden-variable models, and both in the relational and the probabilistic setting, can readily be understood
and presented as teams. This means that formulae in dependence and independence logic can be
directly evaluated on the models we are interested in.
Moreover, it turns out that operators on the team side are naturally compatible with
the structure of hidden-variable models -- for instance, the connection between
probabilistic and relational models corresponds directly to the relationship between
relational and probabilistic team semantics.

\medskip
Although our study does not establish new results on quantum mechanical hidden-variable models as such, we
think that our spelling out of the connections between logics with team semantics and hidden-variable models
in detail is useful and provides the following contributions:
\begin{itemize}
\item We show that the properties of empirical and hidden-variable models can be defined in a very concise
and elegant way in logics of dependence and independence. This also helps to make implicit assumptions in the definition of such properties explicit.
\item Our team-semantical framework enables a full unification of the probabilistic and relational theory. In fact, 
for all the properties that we study, the same formula can be used for both settings, evaluated over relational
teams in the first case, and over probabilistic teams in the second case. This
provides a very elegant perspective on the connection between relational and
probabilistic hidden-variables, complementing earlier work on the connection
between the two such as \cite{AbramskyBra11}, which used a more algebraic approach
based on distributions over semirings.
\item Connections between different properties of hidden-variable models can be
  formulated in terms of \emph{logical entailment} between the team-semantical
  formulae defining these properties and can thus been proved on a purely
  logical level. This may also be beneficial for formalising such proofs in an
  appropriate proof system or theorem prover.
\item Similarly, the existence of an empirically equivalent  hidden-variable model 
that satisfies some combination of desirable properties corresponds to the
\emph{satisfiability} of a suitable
formula by an extension of the team which represents a given empirical model. 
\item On a purely logical level, we establish connections between probabilistic entailment and relational
entailment of formulae from dependence and independence logic, which we believe to be of 
independent interest. Applying these connections to the entailment and satisfiability problems
related to properties of hidden-variable models, we often need only one proof to establish
both the probabilistic and the relational case because, depending on purely syntactic criteria,
one case implies the other. This provides more general reasons why the relational variant of
hidden-variable models is so closely related to their probabilistic counterpart, and further motivates
Abramsky's translation as a special case of a more general framework.
\item We further show that the famous `no-go' theorems by Bell and Kochen-Specker can be formulated
in a natural way in the team semantical framework. In particular we provide logical 
variants of the Kochen-Specker Theorem, highlighting the aspect of contextuality in a 
different way compared to the classical formulation, and discuss the related notion of
\emph{non-contextual choice}.  
\end{itemize}

\medskip By different methods, not directly related to 
team semantics, but based on sheaf theory,
a full unification of the relational and probabilistic theory has already been obtained in
\cite{AbramskyBra11}. This work further gives a unified treatment of non-locality and contextuality
in a very general setting. It remains to be seen whether meaningful connections between
this approach and logics with team semantics can be established.

\medskip\noindent{\bf Acknowledgments. }
Researchers studying logics of dependence and independence have been aware 
for some time of the possibility to use these logics for reasoning about hidden-variable models
in quantum mechanics, and this idea was informally discussed in this research community 
at several occasions. However, when we started our work, no systematic study in this direction
had been conducted yet. Only at a point where this paper was almost finished, it was brought to our 
attention that Joni Puljujärvi and Jouko Väänänen at the University of Helsinki 
have simultaneously pursued a similar line of research, and we had interesting discussions with them
on this topic. Meanwhile they have written a 
joint paper with Samson Abramsky on this work \cite{AbramskyPulVaa21}.
We also wish to thank an anonymous referee for valuable comments and additional references that helped
us to improve this paper. 

\section{Hidden-Variable Models and Teams}\label{sec:models}

We introduce the setup required to understand both relational
and probabilistic hidden-variable models. We also define teams and explain how
hidden-variable models can be cast as teams. Finally, we show that the structure
of hidden-variable models is compatible with operators on the team side and
thereby demonstrate that teams are very well suited to express hidden-variable
phenomena.

\subsection{Relational Models and Teams}

A purely relational setting to speak about hidden-variable models was introduced
by Abramsky \cite{Abramsky13}.

\begin{definition}\label{def:relEmpModels}
Let $M_1,\dots,M_n$ and $O_1,\dots,O_n$ be finite sets.
  We set $M = \prod_{i=1}^n M_i$ and $O = \prod_{i=1}^n
  O_i$. An arbitrary relation $e \sub M \times O$ is called an empirical model
  over $(M,O)$. In this setting, $n$ is called the arity of the system, $M$ the measurement
  set, and $O$ set of outcomes. 
\end{definition}

We usually interpret such a system as having $n$ components which may but need
not be separated in space. We interpret $M_i$ as the set of measurements that
can be performed in component $i$ and $O_i$ as potential outcomes in that
component. The measurement-outcome pairs $(\bar{m},\bar{o}) \in e$ are
interpreted to be \emph{possible} in the model. To make sure that the underlying sets $M$ and $O$ can be 
inferred from the empirical model, we tacitly assume that every value $m_i\in M_i$ and $o_j\in O_j$
appears in at least one tuple of the relation $e$; in database terms this means that 
$\bigcup_{i\leq n} M_i\cup \bigcup_{i\leq n} O_i$ coincides with the active domain of $e$.

\begin{example}\label{example:empirical}
  Let us consider a system with 2 components, called Alice and  Bob. 
  Let $M_1 = \set{a_1,a_2}$ and $M_2 = \set{b_1}$. Thus, Alice has a choice between two 
  measurements while Bob can perform
  only one. In this  example, all measurements reveal a single bit of information,
  i.e.~we let $O_1 = O_2 = \set{+,-}$. An example of an empirical model over
  $(M,O)$ is 
  
  \begin{center}
    \begin{tabular}{c | c c c c}
      $e$ & $(+,+)$ & $(+,-)$ & $(-,+)$ & $(-,-)$ \\ \hline
      $(a_1,b_1)$ & 1 & 0 & 0 & 1 \\
      $(a_2,b_1)$ & 0 & 1 & 1 & 0 
    \end{tabular}
\end{center}

 In this example, the outcome
  of Alice's measurement is always identical to Bob's if Alice chooses $a_1$ and
  always opposite to Bob's if she chooses $a_2$.
\end{example}

\begin{definition}\label{def:relHidModels}
Let $M$ and $O$ be as in \cref{def:relEmpModels}, and let $\Lambda$ be a finite
  set. A \emph{hidden-variable model} over $(M,O,\Lambda)$ is a relation $h \sub M \times O
  \times \Lambda$. The elements $\lambda \in \Lambda$ are called hidden-variables.
\end{definition}

The interpretation of hidden-variable models is similar to that of empirical
models. The hidden-variables are assumed to be unobservable parameters of the
system that may influence outcomes. A triplet $(\bar{m},\bar{o},\lambda) \in h$
means that the measurement-outcome pair $(\bar{m},\bar{o})$ is possible in the
system when the hidden-variable takes the value $\lambda$. 
Again, we tacitly assume that every $\lambda\in\Lambda$ appears in
at least one tuple of $h$. The original
motivation behind the introduction of hidden-variables was that seemingly
unintuitive phenomena on the empirical side in quantum mechanics might be
explained by ignorance about the hidden-variables and not by an intrinsically
non-classical system. However, as we shall see later,
certain non-classical phenomena remain even if hidden-variables are introduced.

\medskip

Hidden-variable models induce empirical models by projecting back onto the
empirically observable parameters $M \times O$, giving rise to a notion of
empirical equivalence, i.e.~a notion of which systems we can distinguish by
experiment and which not.

\begin{definition}\label{def:inducedEmpirical}
  Let $h$ be a hidden-variable model over $(M,O,\Lambda)$. We call
\[ e := \set{(\bar{m},\bar{o}) \in M \times O:  (\exists \lambda \in \Lambda)
 (\bar{m},\bar{o},\lambda) \in h}\]
its \emph{induced empirical model}, and say that $h$ and $e$ \emph{empirically equivalent}.
\end{definition}

Now that we have defined relational models, we present them as teams.

\begin{definition}
A team is a set $X$ of
assignments $s:D\ra A$ with a common finite domain $D = \dom(X)$ of variables
and values in a set $A$. For a tuple of variables
$\bar x=(x_1,\dots, x_m) \in X$, we write $X(\bar x) := \set{(s(x_1),\dots,
  s(x_m))\colon s \in X} \sub A^m$ for the set of values of $\bar x$ in $X$. 
Thinking of an arbitrary but fixed
enumeration of the finite domain of a team $X$ as $\dom(X) = \{x_1,\dots,x_k\}$,
we often identify $X$ with its \emph{relational encoding} $X(\bar x) = \{ s(\bar
x) \colon s \in X \} \sub A^k$ and assignments $s \in X$ with corresponding tuples.
\end{definition}

With that in mind, the interpretation of hidden-variable models as teams is very natural.
To make sure that the underlying sets of
measurements $M= \prod_{i=1}^n M_i$, outcomes $O = \prod_{i=1}^n O_i$, and hidden-variables
$\Lambda$ can be inferred from the team we again
assume that every $m_i \in M_i, o_j \in O_j$ and $\lambda \in \Lambda$
appears in at least one assignment of the team. Then we have a one-to-one
correspondence between hidden-variable models and teams
(and similarly for empirical models).

\begin{definition}\label{def:modelsTeams}
A team $X$ over variables
$\VarH := \set{m_1,\dots,m_n, \allowbreak o_1,\dots,{o_n},
\lambda}$
induces a hidden-variable model, denoted $h_X$, 
with $M_i := X(m_i)$,  $O_i := X(o_i)$, and  $\Lambda := X(\lambda)$ such
that 
\[   (\bar{a},\bar{b},c) \in h_X\quad :\Iff\quad (\exists s \in X)
    s(\bar{m}) = \bar{a}, s(\bar{o}) = \bar{b}, \text{ and }s(\lambda) = c. \]
We denote by $e_X$ the empirical model that is induced by $h_X$. 
\end{definition}

Analogously, empirical models $e \sub M \times O$ are represented by
teams over $\VarE := \set{m_1,\dots,m_n,o_1,\dots,o_n}$.

\subsection{Probabilistic Models and Probabilistic Teams}

Systems arising from quantum mechanics are probabilistic in nature. Thus, the
literature primarily investigates probabilistic hidden-variable models. A
comprehensive discussion of such models can be found in
\cite{BrandenburgerYan08}. However, the relevant definitions 
of probabilistic models and their properties are presented here
in a somewhat different way, on the basis of our 
team semantical framework.

\begin{definition}
  Let $(M,O)$ be as in \cref{def:relEmpModels}. A probability distribution
  $e_{\Pr}\colon M \times O \to [0,1]$
  is called a probabilistic empirical model over $(M,O)$.
Analogously, a probabilistic hidden-variable models is a probability
distributions $h_\Pr$ over $M \times O \times \Lambda$. 
We make use of standard notation
for marginalization and conditionals when speaking about probabilistic models.
For example, for a  probabilistic
hidden-variable model $h_{\Pr}\colon M \times O \times \Lambda \to [0,1]$,
we call the marginalization
$e_{\Pr}\colon M \times O \to [0,1]$ with $e_{\Pr}(\bar{m},\bar{o}) :=
h_{\Pr}(\bar{m},\bar{o}) = \sum_{\lambda \in
  \Lambda}h_{\Pr}(\bar{m},\bar{o},\lambda)$ its induced empirical model. In this
case, $h_{\Pr}$ and $e_{\Pr}$ are called empirically equivalent.
\end{definition}

  The intuition behind empirical equivalence is that empirically equivalent
  models cannot be distinguished by experiment as they agree on the observable
  parameters of the system. For this purpose, the conditional distributions
  $h_{\Pr}(\bar{o} \mid \bar{m})$ are actually slightly more relevant than the
  joint probabilities $h_{\Pr}(\bar{o},\bar{m})$ as they define the outcome
  distributions for fixed measurements; $h_{\Pr}(\bar{m})$ has no meaningful
  interpretation without going into discussions about the experimenters' free
  will. Thus, the literature regards a slightly different notion of empirical
  equivalence. \cite{BrandenburgerYan08} defines an empirical model $e_{\Pr}$ and
  a hidden-variable model $h_{\Pr}$ to be empirically equivalent if they agree
  on all suitable conditional probabilities, i.e.~if $h_{\Pr}(\bar{o} \mid
  \bar{m}) = e_{\Pr}(\bar{o} \mid \bar{m})$ always holds.
  We think that the definition we have chosen, i.e.~requiring the
  slightly stronger $h_{\Pr}(\bar{o},\bar{m}) = e_{\Pr}(\bar{o},\bar{m})$, is
  more natural and elegant for our purposes since empirical equivalence reduces
  to marginal equivalence in the standard probability theoretic sense. 
  For all relevant results it makes no difference which
  choice is made. The probability distribution
  $h_{\Pr}(\bar{m},\bar{o},\lambda)$ can be decomposed into $h_{\Pr}(\bar{m}),
  h_{\Pr}(\lambda \mid \bar{m})$ and $h_{\Pr}(\bar{o} \mid \bar{m},\lambda)$.
  For all properties of interest only the latter two components matter. Thus, we
  could adjust $h_{\Pr}(\bar{m})$ to enforce agreement with the empirical model
  without harming properties of $h_{\Pr}$. In particular, all existence and
  non-existence results which we will show later in this work are independent of
  this technical detail.

  \medskip The connection between probabilistic and relational models is one of
  possibilistic collapse, i.e.~we consider the set of tuples with probability
  greater than zero. This corresponds to our understanding of relational models
  where we interpret $(\bar{m},\bar{o}) \in e$ as statement ``it is
  \emph{possible} in $e$ to obtain the measurement-outcome pair
  $(\bar{m},\bar{o})$''.
 
  \begin{definition}
  A probabilistic empirical model $e_{\Pr}$ over $(M,O)$
  induces the relational model $e := \set{(\bar{m},\bar{o}) \in M \times O :
    e_{\Pr}(\bar{m},\bar{o}) > 0}$. Analogously, probabilistic hidden-variable
  models induce relational hidden variable models.
\end{definition}

Again, we cast probabilistic models as teams, using the notion
of probabilistic teams, which have been considered for instance in 
\cite{DurandHanKonMeiVir18a,HannulaHirKonKulVir19,HannulaKonBusVir20}.

\begin{definition}
A \emph{probabilistic team} is a pair $\XX = (X,\Pr_{\XX})$, where
$X$ is a relational team and $\Pr_{\XX}\colon X
  \to (0,1]$ is a probability distribution over $X$. 
  $X$ is referred to as the underlying team of
  $\XX$. We denote $\Pr$ instead of $\Pr_{\XX}$ if the context is clear.
\end{definition}

As in the relational case, there is a one-to-one correspondence between
probabilistic hidden-variable models and probabilistic teams.

\begin{definition}\label{def:probModelsTeams}
A probabilistic team $\XX = (X,\Pr)$ over
variables $\VarH = \set{m_1, \dots, \allowbreak m_n, \allowbreak o_1, \allowbreak \dots,{o_n},
\lambda}$
induces a hidden-variable model $h_{\XX}$ 
with $M_i := X(m_i)$, $O_i := X(o_i)$, and $\Lambda := X(\lambda)$, 
such that, for every $\bar a\in M$, $\bar b\in O$, and $c\in\Lambda$
\[ h_{\XX}(\bar{a},\bar{b},c) := \begin{cases} \Pr(s)  &\text{ if }s(\bar{m},\bar o,\lambda)=(\bar{a},\bar{b},c)\\
                                                                      0 &\text{ if } (\bar{a},\bar{b},c)\not\in X(\bar m,\bar o,\lambda)
                                                \end{cases}\]
We denote the underlying probabilistic empirical model by $e_{\XX}$.
\end{definition}

\subsection{Compatibility of Models and Teams}
Team semantics admits elegant formulations of hidden-variable phenomena
because the structure of hidden-variable and empirical models nicely interplays with
operators on the team side. First, we observe the following relationship between
probabilistic and relational hidden-variable models.
\begin{lemma}
  Let $\XX = (X,\Pr)$ be a probabilistic team over $\VarH$ and let $h_{\XX}$
  be the corresponding hidden-variable model. Then $X$ is the team representation
  of the induced relational model $h$
  of $h_{\XX}$.
\end{lemma}

Indeed,  $X$ is the underlying support team of $\XX$, containing all assignments with
  non-zero probability, and this corresponds to the notion of
  possibilistic collapse. This holds of course also for empirical models. Further, we observe that
the induced empirical model of a hidden-variable model corresponds to the
restriction to $\VarE$ on the team side.

\begin{lemma}\label{prop:relEmpProject}
  Let $X$ be a relational team over $\VarH$, let $h_X$ be its corresponding hidden-variable
  model and $e_X$ its induced empirical model. Then, $X \upharpoonright
  \VarE$ is the team representation of $e_X$. 
  Similarly, let $\XX$ be a probabilistic team over $\VarH$ representing the
  probabilistic hidden-variable
  model  $h_{\XX}$, and let $e_{\XX}$ be its  induced empirical model. Then, $\XX \upharpoonright
  \VarE$ is the team representation of $e_{\XX}$.
\end{lemma}

This allow us to cast
properties of the underlying model as properties of the hidden-variable model.
We can also formulate sort of a reverse.

\begin{cor}
  Let $X$ be a team over $\VarE$ with corresponding model $e_X$. A team $Y$ over
  $\VarH$ represents an empirically equivalent hidden-variable model if, and only
  if, $Y \upharpoonright \VarE = X$. Analogously, 
  let $\XX$ be a probabilistic team over $\VarE$. A team $\YY$ over
  $\VarH$ represents an empirically equivalent hidden-variable model if, and only
  if, $\YY \upharpoonright \VarE = \XX$.
\end{cor}

All these correspondences are neatly summarized by the commutative
diagram in \cref{fig:commutative}.
\begin{figure}[h]
\[\begin{tikzcd}[row sep={40,between origins}, column sep={40,between origins}]
      & \XX \ar{rr}\ar{dd}\ar{dl} & & \XX \upharpoonright \VarE \ar{dd}\ar{dl} \\
    X \ar[crossing over]{rr} \ar{dd} & & X \upharpoonright \VarE \\
      & h_{\XX} \ar{rr} \ar{dl} & &  e_{\XX} \ar{dl} \\
    h_X \ar{rr} && e_X \ar[from=uu,crossing over]
\end{tikzcd}\]
\caption{Compatibility of team and model structure}\label{fig:commutative}
\end{figure}
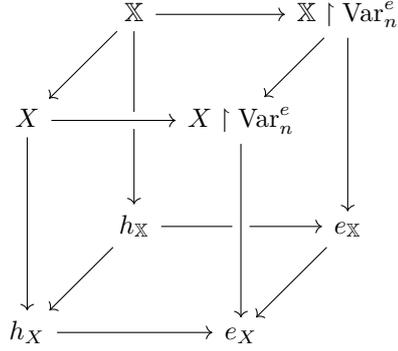

\section{Relational and Probabilistic Team Semantics}

We now introduce the logical machinery that we need to reason about
teams and the hidden-variable models that they represent. 
We first recall the main definitions for logics with relational and
probabilistic team semantics. Then, we develop novel correspondences between the two
frameworks that will then allow us to
unify relational and probabilistic hidden-variables models and 
enable us to develop their theory in parallel.

\subsection{Relational Team Semantics}

Modern logics of dependence and independence are based on atomic properties of teams and
on an interpretation of the classical logical operators $\land$, $\lor$, $\exists$ 
and $\forall$ in the framework of teams.
There are many atomic team properties that have been studied in the context
of such logics. In this paper, we shall need only dependence,
inclusion, and independence, which are defined as follows, for any team $X$,
and for tuples $\bar x,\bar y,\bar z$ of variables in the domain of $X$.

\begin{description}
 \item[Dependence:] $X\models \dep(\tx, \ty)$ if for all $s,s'\in X$ such that $s(\tx) = s'(\tx)$, also $s(\ty) = s'(\ty)$; 
 \item[Inclusion:] $X \models \tx\subseteq\ty$ if $X(\tx)\subseteq X(\ty)$; 
 \item[Independence:] $X \models  \tx\perp_{\tz} \ty$  if for all $s,s'\in X$ such that $s(\tz)=s'(\tz)$, there exists
 some $s''\in X$ such that $s''(\tz)=s(\tz), s''(\tx)=s(\tx)$, and $s''(\ty)=s'(\ty)$.  
\end{description}

Thus $\dep(\tx,\ty)$ describes functional dependence of $\ty$ from $\tx$, in the sense that there
exist a function $f\colon X(\tx)\ra X(\ty)$ such that $s(\ty)=f(s(\tx))$ for every assignment $s\in X$.
Independence however is more than the absence of dependence. The atom  $\tx\perp_{\tz} \ty$ 
expresses that for any fixed value for $\tz$ in the team, the values of $\tx$ and $\ty$ are completely 
independent in the sense that additional information about the value of $\tx$ does not constrain
the possible values of $\ty$ in any way, and vice versa. This variant of independence is
sometimes called \emph{conditional independence}. A \emph{simple independence atom} instead
has the form $\tx\perp\ty$. We then have that $X\models\tx\perp\ty$ if, and only if,
$X(\tx\ty)=X(\tx)\times X(\ty)$, so any value for $\tx$ in the team co-exists with any
value for $\ty$ in a common assignment in $X$.  

\medskip
Syntactically, dependence logic $\FO(\dep)$ is built from dependence atoms $\dep(\tx,\ty)$
and first-order literals, i.e.~atoms and negated atoms of a fixed vocabulary $\tau$,
by the operators $\land$, $\lor$, $\exists$ 
and $\forall$. All formulae are written in negation normal form and
negation is only applied to first-order atoms, not to dependence atoms.
Independence logic $\FO(\perp)$ is built analogously, using independence atoms
$\tx\perp_{\tz} \ty$ instead. 

\medskip
We now present the relational team semantics of  $\FO(\dep)$ and $\FO(\perp)$.
Recall that the traditional Tarski semantics for first-order formulae $\phi(\bar x)$ 
is based on single assignments $s$ whose domain must comprise the variables 
in $\free(\phi)$; we write $\frA\models \phi[s]$ for saying that $\frA$ satisfies
$\phi$ with the assignment $s$.

\medskip
The \emph{team semantics} for a formula $\phi(\bar x)$ in
a logic of dependence and independence 
is instead defined by inductive clauses for the satisfaction relation
$\frA \models_X \phi$
saying that the team $X$ satisfies $\phi$ in $\frA$. 
Since the underlying structure $\frA$ is not of interest for
the purposes of this paper, we shall simplify notation and simply
write $X\models\phi$ instead. 
For a classical first-order literal, i.e.~an atom or its negation,
we simply say that 
$X\models \phi$ if $\phi[s]$ is true in the underlying structure for
all $s \in X$. 
To present the inductive rules of team semantics,
we shall need two basic operations that (possibly) extend
the domain of a given team $X$ (with values in $A$)
to new variables. 

\begin{definition}
Given an assignment $s$, a variable $x$, and a value $a\in A$,
we write $s[x\mapsto a]$ for the assignment that
extends, or updates, $s$ by mapping $x$ to $a$ (and
leaving the values of all other variables unchanged).
The unrestricted \emph{generalisation} of $X$ over $A$ is
$X[x\mapsto A]:=  \{ s[x\mapsto a] \colon s \in X, a \in A \}$,
and the \emph{Skolem-extension} for a function 
$F \colon X \rightarrow \mathcal{P}(A)\setminus \{ \emptyset \}$ is
$X[x\mapsto F] := \{ s[x\mapsto a] \colon s \in X, a \in F(s) \}$. 
\end{definition}

The following inductive rules then extend the team semantics of atomic team
properties and first-order literals to arbitrary formulae in $\FO(\dep)$
or $\FO(\perp)$:
\begin{itemize}
\item
$X\models \phi_1\wedge \phi_2$ if $X\models \phi_i$ for
$i = 1,2$;
\item $X\models \phi_1\vee \phi_2$ if $X = X_1 \cup X_2$ for two
teams $X_i$ such that  $\frA\models_{X_i} \phi_i$;
\item
$X\models \forall x \phi$ if $X[x\mapsto A]\models\phi$; 
\item
$X\models \exists x \phi$ if 
$X[x\mapsto F]\models\phi$, for some suitable Skolem extension of $X$ by
 a function $F \colon X \rightarrow \mathcal{P}(A)\setminus \{ \emptyset \}$. 
\end{itemize}

Since $\dep(\tx,\ty)\equiv \ty\perp_{\tx} \ty$, we can
freely use dependence atoms in $\FO(\perp)$, and it is known \cite{Galliani12}
that inclusion atoms $\tx\subseteq\ty$ are definable in  $\FO(\perp)$ as well.
In fact, the logic $\FO(\dep,\subseteq)$ with both inclusion and dependence
atoms is equivalent to $\FO(\perp)$ and it suffices to use simple independence atoms
to get the full power of $\FO(\perp)$.

\subsection{Probabilistic Team Semantics}
To understand probabilistic hidden-variable phenomena, we make use of a
probabilistic variant of team semantics. Our definitions are 
essentially the same as in 
\cite{DurandHanKonMeiVir18a,HannulaHirKonKulVir19,HannulaKonBusVir20} with deviations 
in some details. However, we stress the fact that we consciously
select notations that invite jumping back and forth between probabilistic
and relational semantics.
In particular, we use precisely the same syntax for the probabilistic variants of  
$\FO(\dep)$ and $\FO(\perp)$, but we evaluate
the formulae now over probabilistic teams $\XX = (X,\Pr_{\XX})$.
To define the probabilistic semantics for these logics, we first have to
describe the meaning of first-order literals
and of dependence and independence atoms.

\begin{itemize}
\item For first-order literals $\phi$, we let $\XX \models \phi$ if, and only if,
      $X\models \phi$ in the relational sense.
\item ${\XX} \models \dep(\bar{x},\bar{y})$ if for all $\bar{a} \in
      X(\bar{x})$ there is a $\bar{b} \in X(\bar{y})$ such that 
      $\Pr(\bar{y} = \bar{b} \mid \bar{x} = \bar{a}) = 1$.
      Note that this corresponds to a deterministic sense of dependence.
\item ${\XX} \models \bar{x} \perp_{\bar{z}} \bar{y}$ if we have conditional
      stochastic independence between $\bar{x}$ and $\bar{y}$ for any given $\bar{z}$.
      Formally, this means that for all $\bar{a} \in X(\bar{x}), \bar{b} \in X(\bar{y}), \bar{c} \in
      X(\bar{z})$,
  \[ \Pr(\bar{x} = \bar{a}, \bar{y} = \bar{b} \mid \bar{z} = \bar{c})
        = \Pr(\bar{x} = \bar{a} \mid \bar{z} = \bar{c})\cdot \Pr(\bar{y} = \bar{b} \mid \bar{z} = \bar{c}). \]
\end{itemize} 

The probabilistic team semantics of the logical operators generalises the relational
semantics. To make this precise, we have to give appropriate definitions for the
split, the generalisation, and the Skolem extension of a probabilistic team.

\begin{definition}
Let $\XX=(X,\Pr_{\XX})$ be a probabilistic team with $A$ as its set of values,
and  let $\Delta(A)$ denote the set of all probability distributions over $A$. Consider a function 
$\Ff\colon X \to \Delta(A)$ that maps every $s\in X$ to a probability
distribution $\Ff_s\in \Delta(A)$ on $A$.
It induces a function $F\colon X \to \mathcal{P}(A)\setminus \{ \emptyset \}$ via the support of the distributions $\Ff_s$, i.e.~by setting $F(s) := \set{a \in A: \Ff_s(a) > 0}$. Then, the Skolem extension $\XX[x\mapsto\Ff]$ is
defined as the probabilistic team over $X[x\mapsto F]$ with distribution
  \[ \Pr(s[x\mapsto a]) := \sum_{\substack{t \in X, \\ t[x\mapsto a] = s[x\mapsto a]}}
    \Pr_{\XX}(t)\cdot F_t(a). \] 
Note that the right-hand side simplifies to $\Pr_{\XX}(s)\cdot F_s(a)$ if $x$ is a new variable.
\end{definition}

Intuitively, this means that the probability mass of $s$ is split over
multiple extensions $s[x\mapsto a]$ where the proportion assigned to each $a$ is
  given by $F_s(a)$. Thus, we can interpret $F_s(a)$ as the probability of 
  the event $x = a$ conditional on agreement with $s$ on $\dom(s)$. 
  The function $F$ tells us how we have to extend our underlying relational
  support team.

  \medskip
Note that the marginalization of Skolem extensions by $\Ff$ gives back the
original team. More formally, for $\XX = (X,\Pr)$ with $x \not\in\dom(X)$ and
  a Skolem extension $\YY = \XX[x\mapsto\Ff]$, it holds that $\XX = \YY \upharpoonright
  \dom(X)$.

\begin{definition}
 The \emph{uniform extension} $\XX[x\mapsto A]$ of $\XX$ is the special case of a Skolem extension
  $\XX[x\mapsto \Ff]$ where $\Ff$ maps all $s \in X$ to the uniform distribution over $A$.
  Explicitly, we get the distribution
  \[ \Pr(s[x\mapsto a]) = \frac{1}{\abs{A}}\sum_{\substack{t\in X \\ t[x\mapsto a] = s[x\mapsto a]}}
    \Pr_{\XX}(t) . \] 
 Again this simplifies  to $\Pr(s[x\mapsto a]) =\Pr_{\XX}(s)/ {\abs{A}}$ if $x$ is a new variable. This corresponds to a
uniform split of the probability mass of $s$ over all $\abs{A}$ extensions.
\end{definition}

\begin{definition}
 Let $\XX$ be a probabilistic team with $X$ as its underlying team. We now define the
 probabilistic team semantics of the logical operators as follows:
 \begin{itemize}
 \item $\XX \models \phi_1 \land \phi_2$ if $\XX \models \phi_i$ for $i=1,2$.
  Thus, conjunctions are defined in the straightforward way, as in relational team semantics.
    \item $\XX \models \phi_1 \lor \phi_2$ if there are
      probabilistic teams $\XX_1, \XX_2$ and $\lambda \in [0,1]$ with
      $X = X_1\cup X_2$ and $\Pr_{\XX} = (1-\lambda)\Pr_{\mathbb{X}_1} + \lambda\Pr_{\mathbb{X}_2}$
      such that $\XX_i \models  \phi_i$ for $i=1,2$.
      In other words, this means that we can split $\XX$ into a convex
      combination of two teams that satisfy the formulae $\phi_i$.
\item $\XX \models \forall x \phi$ if $\XX[x\mapsto A]\models \phi$, i.e.~if the uniform extension 
of $\XX$ satisfies $\phi$.
\item $\XX \models \exists x \phi$ if $\XX[x\mapsto\Ff]\models\phi$
for some $\Ff\colon X \to \Delta(A)$, i.e.~if a suitable Skolem extension of $\XX$ satisfies $\phi$.
\end{itemize}
\end{definition}

The existential quantifier is most relevant for our purposes and
we want to provide more intuition for it.
For a variable $x \not\in \dom(X)$, $\XX\models \exists x \phi$ expresses the
existence of another probabilistic team $\YY$ over $\dom(X) \cup \set{x}$ such that
  \begin{enumerate}
   \item $\YY$ satisfies $\phi$, and
    \item $\YY$ marginalizes to $\XX$ if restricted to $\dom(X)$. This means that for
  all $s \in X$ we have that
  $\Pr_{\XX}(s) = \sum_{a \in A} \Pr_{\YY}(s[x\mapsto a])$.
\end{enumerate}

\subsection{Back and Forth between Relational and Probabilistic Teams}\label{subsec:compareTeams}

We now compare probabilistic and relational team semantics. We
start with the satisfaction relation.  Note that it is essential for the formulation of this
theorem that we use the same syntax for both variants.

\begin{theorem}\label{theorem:compareSemantics}
Let $\XX$ be a probabilistic team and $X$ its underlying team. 
For every formula $\psi\in\FO(\perp)$, we have that $\XX \models \psi$ implies $X\models \psi$.
For every $\phi\in\FO(\dep)$, we also have the converse, so that 
 $\XX \models \phi$ if, and only if, $X\models \phi$.
 \end{theorem}
 
\begin{proof} We proceed by induction.
\begin{itemize}
\item For first-order literals $\phi$, we obviously have that $\XX \models \phi$ if, and only if, $X\models \phi$.
\item Assume that $\XX \models \bar{x} \perp_{\bar{z}} \bar{y}$. Now let
  $s, s' \in X$ with $s(\bar{z}) = s'(\bar{z}) = \bar{c}$. With $\bar{a} =
  s(\bar{x})$ and $\bar{b} = s'(\bar{y})$ we get that
  \[ \Pr(\bar{x} = \bar{a}, \bar{y} = \bar{b} \mid \bar{z} = \bar{c}) =
    \underbrace{\Pr(\bar{x} = \bar{a} \mid \bar{z} =
      \bar{c})}_{\geq\Pr(s)>0} \cdot \underbrace{\Pr(\bar{y} = \bar{b} \mid
      \bar{z} = \bar{c}}_{\geq\Pr(s') > 0}) > 0. \]
  Since $X$ is the support team of $\XX$, there exists $s''\in X$ with $s''(\bar{x}, \bar{y}, \bar{z}) = (\bar{a},\bar{b},\bar{c})$
  which proves that $X \models \bar{x} \perp_{\bar{z}} \bar{y}$.
 \item For dependence atom, the implication from left to right follows from the fact that dependence atoms can
  be understood as special cases of independence atoms, the other direction
  is straightforward since $X$ is defined to be the support of $\Pr$.
\end{itemize}  
It remains to be shown that the logical operators preserve the implications in both directions.
This is obvious for conjunctions.
\begin{itemize}
\item For $\XX\models \psi \lor \vartheta$, we choose $\YY \models \psi$ and $\ZZ \models
  \vartheta$ with $\XX = (1-\lambda)\YY + \lambda\ZZ$ so that $X
  = Y \cup Z$. By induction hypothesis, $Y \models \psi$ and $Z
  \models \vartheta$, hence $X\models\psi\lor\vartheta$. Conversely, given a probabilistic team
  $\XX$ and a decomposition $X=Y\cup Z$ of the underlying team with $Y\models\psi$ and
  $Z\models\vartheta$, it is straightforward to construct  $\YY$ and $\ZZ$
  with underlying $Y$ and $Z$ such that 
  $\Pr_{\XX} = (1-\lambda)\Pr_{\YY} + \lambda \Pr_{\ZZ}$ for some $\lambda \in [0,1]$.
  By induction hypothesis $\YY\models\psi$ and
  $\ZZ\models\vartheta$, hence $\XX\models\psi\lor\vartheta$. 
  
\item For  formulae $\forall x \psi$, the induction step follows from the fact that 
the underlying team of $\XX[x\mapsto A]$ is $X[x\mapsto A]$. 
  
\item If $\XX\models \exists x \psi$ by means of  $\Ff\colon X \to \mathcal{P}(A)\setminus \{ \emptyset \}$, 
  we put $F(s)= \set{a \in A: \Ff_s(a) > 0}$ to get the Skolem extension $X[x\mapsto F]$ as
  underlying team of $\XX[x\mapsto \Ff]$, so $X[x\mapsto F]\models\psi$ and hence $X\models\E x\psi$.
Conversely, let $X\models \exists x \psi$, hence $X[x\mapsto F]\models\psi$
for some $F\colon X \to \mathcal{P}(A)\setminus \{ \emptyset \}$. Let $\Ff: X\mapsto \Delta(A)$ be the function that maps
every $s\in X$ to the uniform distribution over $F(s)$.
Then, $\XX[x\mapsto\Ff]$ has $X[x\mapsto F]$ as underlying team and the induction hypothesis
is applicable.\qedhere
\end{itemize}
\end{proof}
We next address the matter of logical entailment. 
\begin{definition}
  Let $\psi, \phi\in\FO(\perp)$. We write $\psi \relEnt
  \ph$ if for all suitable relational teams $X$ with $X \models \psi$ also $X
  \models \ph$ holds. Analogously, we define $\psi \probEnt \ph$ as entailment
  according to probabilistic team semantics. If both $\psi \relEnt \ph$ and
  $\psi \probEnt \ph$ hold, we write $\psi \allEnt \ph$.
\end{definition}

Results that are analogous to the theorem below have been established for databases but 
apparently not applied to teams so far.

\begin{theorem}\label{theorem:compareEntailment}\mbox{}
  \begin{enumerate}
    \item For formulae in $\FO(\dep)$, $\relEnt$ and $\probEnt$ coincide.
    \item If $\psi \in \FO(\dep), \ph \in \FO(\perp)$, then $\psi \probEnt \ph
      \impl \psi \relEnt \ph$.
    \item If $\psi \in \FO(\perp)$ and $\ph \in \FO(\dep)$, then $\psi \relEnt
      \ph \impl \psi \probEnt \ph$.
    \item For $\ph, \psi \in \FO(\perp)$, neither implication holds in
      general. Conjunctions of conditional independence atoms suffice
      to generate counterexamples.
  \end{enumerate}
\end{theorem}

\begin{proof}
  Notice that (2) and (3) follow immediately from
  \cref{theorem:compareSemantics} and that (1) follows by combining (2) and (3)
  since $\FO(\dep) \sub \FO(\perp)$. (4) remains to be shown. We adapt
  formulae and counterexamples from
  \cite{WongBut00}.

  \medskip We first show that, in general, $\psi \relEnt \ph$ does not imply $\nimpl \psi
  \probEnt \ph$. Let $\psi_1 = z \perp_{x} w \land z \perp_{y} w \land x
  \perp_{wz} y$ and $\ph_1 = z \perp_{xy} w$.
The
  intuition behind $\psi_1 \relEnt \psi_2$ is that we can combine $z \perp_x w$
  and $z \perp_y w$ to get $z \perp_{xy} w$ by using $x \perp_{wz} y$ to enforce
  agreement on $xy$.
  
  \medskip
  Let $X \models \psi_1$ and $s_1,s_2\in X$ with $(a,b)=s_1(x,y) = s_2(x,y)$. Further, let
  $c := s_1(z)$ and $d := s_2(w)$. Since $X \models z \perp_{x} w$, there exists 
  $s_3 \in X$ with $s_3(x) = a$ and $s_3(z,w) = (c,d)$. Similarly, we obtain $s_4 \in
  X$ with $s_4(y) = b$ and $s_4(z,w) = (c,d)$ via $X \models z \perp_y w$. 
  By applying $X \models x \perp_{wz}y$ on $s_3$ and $s_4$, we finally obtain $s_5 \in X$ with 
  $s_5(x,y,z,w) = (a,b,c,d)$
  which shows that $X \models \ph_1$. Thus, $\psi_1 \relEnt \ph_1$. Howver, it is straightforward
  to verify that the probabilistic team given below satisfies
  $\psi_1$ but not $\ph_1$.

\begin{center}
        \begin{tabular}{ c  c  c  c || c}
          $x$ & $y$ & $z$ & $w$ & $\Pr$ \\
          \hline \hline
          0 & 0 & 0 & 0 & 0.2 \\ \hline
          0 & 0 & 0 & 1 & 0.2 \\ \hline
          0 & 0 & 1 & 0 & 0.2 \\ \hline
          0 & 0 & 1 & 1 & 0.1 \\ \hline
          0 & 1 & 1 & 1 & 0.1 \\ \hline
          1 & 0 & 1 & 1 & 0.1 \\ \hline
          1 & 1 & 1 & 1 & 0.1
        \end{tabular}
        \hspace*{2cm}
        \begin{tabular}{ c  c  c  c}
          $x$ & $y$ & $z$ & $w$  \\
          \hline \hline
          0 & 0 & 0 & 0 \\ \hline
          0 & 0 & 0 & 1 \\ \hline
          0 & 1 & 0 & 0 \\ \hline
          1 & 0 & 0 & 0 \\ \hline
          1 & 1 & 0 & 0 \\ \hline
          1 & 1 & 1 & 0  
        \end{tabular}\\[2mm]
      {\hspace*{5mm}$\psi_1 \not\probEnt \ph_1$ \hspace*{30mm}  $\psi_2 \not\relEnt \ph_2$}
      \end{center}
      
\medskip
To prove that $\psi \probEnt \phi$ does not necessarily imply that $\psi \relEnt \ph$, we use the fact from
 \cite[equation (A3), p.~15]{Studeny89} that   $\psi_2 \probEnt \ph_2$, for
$\psi_2 = x \perp_{yz} y \land z \perp_{x} w \land z \perp_{y} w \land x \perp y$ and $\ph_2 = z \perp w$ 
(proven in a different context with techniques employing measure theory
 and information theory). On the other side, a team 
proving that $\psi_2 \not\relEnt \ph_2$ is given above (to the right).
 \end{proof}

We remark that  if $\ph,\psi \in \FO(\perp)$ are conjunctions of
  \emph{simple independence atoms}, then $\psi \relEnt \ph \Iff \psi
  \probEnt \ph$. This follows by observing that the proof calculi provided in
  the literature, see \cite{GallianiVaa13} for the relational and \cite{GeigerPazPea91} for
  the probabilistic case, coincide.

\section{Reasoning about Hidden-Variable Properties in Independence Logic}\label{sec:propModels}
We now define and investigate properties of hidden-variable models by means of
dependence and independence logic.

\subsection{Logical Definitions of Properties of Hidden-Variable Models}\label{subsec:defProperties}
In the literature, many properties of hidden-variable and empirical models are
investigated. Given that these models can be seen as (relational or probabilistic) teams,
we can define properties of such models by logical formulae that are evaluated over such teams.
Actually, since natural properties are defined for models of arbitrary arity, their logical
definition is not given by a single formula, but by a uniform family of such formulae, one for each arity.
That is, for a property $P$ of probabilstic hidden-variable models, we present a family of formulae
$\psi_n\in\FO(\perp)$, with free variables in $\VarH$ such that for any
probabilistic team $\XX$ that represents a hidden-variable model $h_{\XX}$ of arity $n$, we have that
$h_{\XX}$ has property $P$ if, and only if, $\XX\models\psi_n$.
Similarly, for empirical models, and for the relational case.

\medskip
Actually, it is a fundamental observation that for all elementary properties, 
syntactically identical formulae work for the relational and the probabilistic case
simultaneously. Also, our formulae are in fact very simple; essentially just
conjunctions of atoms. This is quite elegant and another indication that team
semantics is very well suited to express hidden-variable phenomena.

\subparagraph*{Properties of empirical models}\label{def:empProperties}

We start with presenting three fundamental properties
that empirical models may or may not have.
Note that in our formulae, the role of $n$ is hidden in notations such as
$\bar{m}$ for $m_1\dots m_n$ and $\bar o$ for $o_1\dots,o_n$.
The following notation is helpful: For a tuple  $\bar z= (z_1\dots z_n)$, let
  $\bar{z}_{-i} = (z_1\dots z_{i-1}z_{i+1}\dots z_n)$.

\begin{description}
\item[Weak Determinism: ] An empirical model is weakly deterministic if the combined
    measurements $\bar{m}$ in all components deterministically determine the
    outcome $\bar{o}$ of all components. However, this property does not forbid that the
    outcome of the $i$-th component depends on the measurement choice in
    components other than $i$. This property is defined by
      \[ \WeakDet^e_n := \dep(\bar{m},\bar{o}).\]
\item[Strong Determinism: ] If we strengthen Weak Determinism and require that the
  outcome of the $i$-th component is uniquely determined by the measurement
  choice in that component, we obtain the notion of Strong Determinism, defined by
\[ \StrongDet^e_n := \bigwedge_{i=1}^n \dep(m_i,o_i).\]
\item[No-Signalling: ] This property is a bit more complicated. Roughly speaking, the idea is to
  formalize that no information can be transmitted to component $i$ by choice of
  measurement in components other than $i$. Thus, the outcome of the $i$-th
  component (which for example Alice receives) may, conditional on her choice of
  measurement, not depend on other measurements. An illustration of this
  point is given in \cref{exampleNS} below. No-Signalling is closely
  related to the famous No-Communication Theorem in Quantum Mechanics. Formally,
\[ \NoSig^e_n := \bigwedge_{i=1}^n   o_i \perp_{m_i} \bar{m}_{-i}.\]
It may not be completely obvious why the formulae $\NoSig^e_n$ have the intended
meaning. Intuitively, an independence atom $\bar x \perp_{\bar z} \bar y$
expresses that conditional on knowing $\bar z$, getting information about $\bar
x$ does not give additional information about $\bar y$. Thus, $\NoSig^e_n$
states that if $m_i$ is known, the other measurements $\bar{m}_{-i}$ do not give
additional information about $o_i$.
\end{description}

\begin{example}\label{exampleNS}
  To illustrate how a violation of No-Signalling allows communication, consider
  the empirical model defined by the following team $X$:
  \begin{center}
        \begin{tabular}{ c  c | c  c}
          $m_1$ & $m_2$ & $o_1$ & $o_2$ \\
          \hline \hline
          $a$ & $b_1$ & $+$ & $-$ \\ \hline
          $a$ & $b_2$ & $-$ & $-$ 
        \end{tabular}
 \end{center} 
 In this case, $X \not\models o_1
  \perp_{m_1} m_2$ which allows Bob to instantly transmit information to Alice
  by choice of his measurement. If Bob makes measurement $b_1$, Alice receives a
  $+$ in her experiment; if he chooses $b_2$, Alice receives a $-$. In this configuration, Bob
  can send a full bit of information to Alice by means of his measurement
  choice; No-Signalling is rightfully violated. Also note that this example satisfies
  Weak Determinism but not Strong Determinism.
\end{example}

It is a relevant feature of team semantics, often called locality,  that the
meaning of a formula depends only on those variables that occur free in it. 
Thus, even if we evaluate a formula over a hidden-variable model, we can
express properties of the underlying empirical model (if the variable $\lambda$
is not used in the formula).
Properties of the underlying empirical model can therefore be cast as
properties of the hidden-variable model itself. This is made precise in the
following proposition.

\begin{prop}\label{prop:probFormulaeEmp}
  Let $\XX$ be a probabilistic team over $\VarH$, and let $\psi^e\in\FO(\perp)$ be a
  formula over $\VarE$ that formalizes a property $P$ of probabilistic empirical models. Then
  $\XX \models \psi^e$ if, and only if, $e_{\XX}$ has property $P$.
\end{prop}
\begin{proof}
By locality, $\XX \models \psi^e$ is equivalent to $\XX
  \upharpoonright \VarE \models \psi^e$. By \cref{prop:relEmpProject}, $\XX
  \upharpoonright \VarE$ is the team representation of $e_\XX$.
  \end{proof}

By the same argument, the analogous statement for relational teams and models also holds.

\subparagraph*{Properties of hidden-variable models}\label{def:hidProperties}

We now show, in a similar fashion, that common properties of hidden-variable models
are definable in dependence and independence logic. The reader may want to compare 
our formula with the definitions in \cite{Abramsky13, BrandenburgerYan08}. 
Except in the case of locality (which we shall discuss separately),
it is straightforward to see that our definitions capture the properties from
the literature. We feel that our definitions by formulae of team semantics are
more compact and precise, and they highlight the essential features better than an 
explicit unrolling of the semantic of dependence and independence atoms in every
instance.

\begin{description}
\item[Weak Determinism: ] For hidden-variable models, Weak Determinism means 
  that for every fixed value
  of a hidden-variable the chosen measurements in all components determine the
  outcomes. It is essential that the outcomes may also depend on the
  hidden-variable, which allows explaining non-deterministic empirical models by
  deterministic hidden-variable models with more fine-grained internal states. Formally
\[ \WeakDet^h_n := \dep(\bar{m}\lambda,\bar{o}).\]
\item[Strong Determinism: ] Strong Determinism strengthens Weak Determinism by requiring that for
  every fixed value of the hidden-variable the measurement in the $i$-th component 
  determines the outcome in the $i$-th component. Formally,
\[ \StrongDet^h_n := \bigwedge_{i=1}^n \dep(m_i\lambda,o_i).\]
\item[Single-Valuedness: ] A hidden-variable model is single valued if the hidden-variable set $\Lambda$
has only one element. Such a model is essentially just an empirical model, because a single-valued
variable does not provide any additional freedom.
\[ \SingVal^h_n := \dep(-,\lambda).\]
\item[$\lambda$-Independence: ] This formalizes the idea that the measurement process and the
  hidden-variable should be independent. It is related to the intuition that
  the physical reality we measure shall be independent of the experimenters'
  choices.
  A particularly pathological counterexample to $\lambda$-Independence would be
  a system where the hidden-variable of the system uniquely determines which
  measurements will be done by the experimenters.
\[ \LInd^h_n := \bar{m} \perp \lambda.\]
\item[Outcome Independence: ] This property states that outcomes of measurements
  in different components shall be independent from each other when conditioning on
  $\bar{m}\lambda$. This means that when we perform the same measurements in the same
  hidden states, the outcomes of the various components shall not provide
  information about each other. 
  Note that violation of this property is related to the phenomenon of quantum
  entanglement where we obtain correlations between outcomes even in case of
  spatially separated particles. 
\[ \OutInd^h_n := \bigwedge_{i=1}^n    o_i  \perp_{\bar{m}\lambda} \bar{o}_{-i}.\] 
\item[Parameter Independence: ] This is essentially the hidden-variable analogue of
  No-Signalling. For fixed $m_i\lambda$, the outcome $o_i$ of component $i$ shall be
  independent of the measurements taken in the other components:
\[\ParInd^h_n := \bigwedge_{i=1}^n o_i \perp_{m_i\lambda} \bar{m}_{-i}.\] 
\item[Locality: ] A hidden-variable model is local if all components are independent of each
  other. Roughly speaking, to understand a system satisfying Locality, one only
  needs to understand every component and piece them together independently.
  Jarrett has shown that this is adequately expressed by the conjunction of Parameter
  Independence and Outcome Independence \cite{Jarrett84}:  
\[ \Loc^h_n:= \OutInd^h_n \land \ParInd^h_n.\]
Heuristically speaking, Outcome Independence gives independence of the $i$-th component of the
other outcomes while Parameter Independence gives independence of the other measurements. Put
together, they give full independence  the other components. This argument is made
precise below.
\end{description}

\begin{example}
  We adapt \cref{exampleNS} to a
  hidden-variable model satisfying Single-Valuedness: 
        \begin{center}
        \begin{tabular}{ c  c | c  c | c}
          $m_1$ & $m_2$ & $o_1$ & $o_2$ &$\lambda$ \\
          \hline \hline
          $a$ & $b_1$ & $+$ & $-$ & $\lambda_1$\\ \hline
          $a$ & $b_2$ & $-$ & $-$ & $\lambda_1$
        \end{tabular}
        \end{center}
The resulting model does satisfy neither Parameter
  Independence nor Strong Determinism but satisfies Weak Determinism and Outcome
  Independence. 
 \end{example}

An intuitive and precise semantic characterization of Locality is shown in the
following two lemmata for the relational and probabilistic case respectively.

\begin{lemma}\label{lem:locSemanticsRel}
  Let $X$ be a relational team of arity $n$ over $\VarH$. $X \models \Loc^h_n$
if, and only if,  for every tuple $(\bar{a},c) \in X(\bar{m},\lambda)$ and for all 
$b_1,\dots,b_n$ with $b_i \in X(o_i)$ the following condition (*) holds:
If there is an assignment $s_i \in X$ with $s_i(m_i,o_i,\lambda) = (a_i,b_i,c)$
for all $i$,
then there is an $s \in X$ with $s(\bar{m},\bar{o},\lambda) = (\bar{a},\bar{b},\bar{c})$.
\end{lemma}

\begin{proof}
Assume that $X \models  \Loc^h_n$, i.e., $X\models \OutInd^h_n \land \ParInd^h_n$. 
Let $(\bar{a},c) \in X(\bar{m},\lambda)$ and $b_i \in X(b_i)$ for every
  $i \leq n$. Let $s_i \in X$ with $s_i(m_i,o_i,\lambda) = (a_i,b_i,c)$. 
  We need to construct some $t \in X$ with $t(\bar{m},\bar{o},\lambda) = (\bar{a},\bar{b},c)$.
  Since $(\bar a, c) \in X(\bar{m},\lambda)$, there exists $\tilde{s} \in X$ with
  $\tilde{s}(\bar{m},\lambda) = (\bar{a},c)$. By applying Parameter Independence on $s_i$ and
  $\tilde{s}$ we get $\tilde{s_i} \in X$ which fulfills
  $\tilde{s_i}(\bar{m},\lambda) = (\bar{a},c)$ and $\tilde{s_i}(o_i) = b_i$. By
  inductively applying Outcome Indpendence, we can now construct a sequence $t_1,\dots,t_n$ with
  $t_i(\bar{m},\lambda) = (\bar{a},c)$ and $t_i(o_1,\dots, o_i) = (b_1,\dots, b_i)$.
  Choosing $t := t_n$ completes the argument for the first implication.

  \medskip For the converse, suppose that $X$ satisfies condition (*). We prove
  that $X \models \OutInd^h_n$. Let $s_1,s_2 \in X$ with $s_1(\bar{m},\lambda) =
  s_2(\bar{m},\lambda)$ and let $i \leq n$. Condition (*) immediately gives that
  $s_3$ given by $s_3(\bar{m},\lambda) = s_1(\bar{m},\lambda)$, $s_3(o_i) =
  s_1(o_i)$, and $s_3(\bar{o}_{-i}) = s_2(\bar{o}_{-i})$ is indeed in $X$; this
  proves Outcome Independence. Parameter Independence is shown similarly.
\end{proof}

Condition (*) in  \cref{lem:locSemanticsRel} is a common alternative definition of
  Locality. In words it essentially states that the question whether a
  measurement-outcome pair $(\bar{a},\bar{b})$ is possible for a fixed value of
  $\lambda$ reduces to whether the measurement-outcome pair $(a_i,b_i)$ is
  possible in every component $i$. However, a slight technical complication is
  that we require this to only hold for measurements $\bar{a}$ that are possible
  for a fixed value of $\lambda$: We explicitly do not require that $(a_i,c) \in
  X(m_i,\lambda)$ for all $i$ implies $(\bar a,c) \in X(\bar{m},\lambda)$. The latter
  condition is called Measurement Locality in \cite{Abramsky13} and rather
  natural but not entailed in the usual definition of Locality.

\begin{example}
  Consider the hidden-variable model given by the following team:
         \begin{center}
        \begin{tabular}{ c  c | c  c | c}
          $m_1$ & $m_2$ & $o_1$ & $o_2$ & $\lambda$ \\
          \hline \hline
          $a$ & $b$ & 1 & 0  & $\lambda_1$ \\ 
          $a$ & $b$ & 1 & 1 & $\lambda_1$ \\ 
          $a$ & $c$ & 0 & 1 & $\lambda_1$ \\ 
          $a$ & $c$ & 0 & 0 & $\lambda_1$ \\ \hline
          $a$ & $b$ & 0 & 0 & $\lambda_2$ \\ 
          $a$ & $c$ & 0 & 1 & $\lambda_2$
        \end{tabular}
\end{center}
This model does not satisfy Locality: For the hidden-variable $\lambda_1$, it is
  possible that $a$ results in 0 and that $b$ results in $1$. However,
  $(ab,01,\lambda_1)$ is not an element of the team even though $(ab,\lambda_1)
  \in X(\bar{m},\lambda)$. Thus, Locality is falsified according to the previous lemma. Note
  that, in this instance, Parameter Independence is violated since the choice of
  measurement in the second component influences the possible results in the
  first component. However, Outcome Independence holds.
  Also, this team satisfies $\lambda$-Independence since all possible
  measurements, i.e.~$ab$ and $ac$, appear for both values of the
  hidden-variable.
\end{example}

In the probabilistic case, Locality means that the probability factors over the
various components. This is in a sense a quantitative analogue of the
qualitative statement from the previous lemma.

\begin{lemma}
  Let $\XX = (X,\Pr)$ be a probabilistic team over $\VarH$. Then  $\XX \models \Loc^h_n$
if, and only if the following condition (**) holds: For all $(\bar{a},c) \in X(\bar{m},\lambda)$ and 
all $b_1,\dots,b_n$ with $b_i \in X(o_i)$ 
      \[ \Pr(\bar{o} = \bar{b} \mid \bar{m} = \bar{a}, \lambda = c) =\prod\nolimits_{i=1}^n
        \Pr(o_i = b_i \mid m_i = a_i, \lambda = c).\]
\end{lemma}
\begin{proof}
Assume that $\XX \models  \Loc^h_n$.
Let $(\bar a, c) \in X(\bar{m},\lambda)$ and $b_i \in X(o_i)$ for
  $i \leq n$. To simplify notation, let $\bar o_k:=o_1\dots o_k$ and $\bar b_k:=b_1\dots b_k$
  We apply induction over $k$ to obtain
  \begin{align*}
  &\quad\Pr(\bar o_{k+1} = \bar b_{k+1} \mid \bar{m} = \bar a, \lambda = c) \\
  &\qquad= \quad \Pr(o_{k+1} = b_{k+1} \mid \bar o_k = \bar b_k, \bar{m} = \bar{a}, \lambda = c)
  \cdot \Pr(\bar o_k = \bar b_k \mid \bar{m} = \bar{a}, \lambda = c) \\
  &\qquad= \quad \Pr(o_{k+1} = b_{k+1} \mid \bar{m} = \bar{a}, \lambda = c) \cdot \prod\nolimits_{i=1}^k \Pr(o_i = b_i \mid \bar{m} = \bar{a}, \lambda = c) \\
 &\qquad= \quad \prod\nolimits_{i=1}^{k+1} \Pr(o_i = b_i \mid \bar{m} = \bar{a}, \lambda = c) = 
  \prod\nolimits_{i=1}^{k+1} \Pr(o_i = b_i \mid m_i = a_i, \lambda = c),
\end{align*}
by applying Outcome Independence, Parameter Independence, and the induction hypothesis.
Setting $k+1 = n$ completes the argument.

\medskip
For the converse, let $\XX$ satisfy condition (**). We first show
that $\XX \models \ParInd^h_n$:
\begin{align*}
  &\Pr(o_i = b_i \mid \bar{m} = \bar{a},\lambda = c) \quad=\quad 
  \sum\nolimits_{\bar{b}_{-i}} \Pr(\bar{o} = \bar{b}, \bar{m} = \bar{a},\lambda = c) \\
  &\quad=\quad \sum\nolimits_{\bar{b}_{-i}} \prod\nolimits_{j=1}^n \Pr(o_j = b_j \mid m_j = a_j, \lambda = c) \\
  &\quad=\quad\Pr(o_i = b_i \mid m_i = a_i, \lambda = c)*\prod\nolimits_{j \neq i} \sum\nolimits_{b_j} \Pr(o_j = b_j \mid m_j = a_j, \lambda = c) \\
  &\quad=\quad\Pr(o_i = b_i \mid m_i = a_i, \lambda = c). 
\end{align*}
Outcome Independence is then shown by the following equation for all $i$:
\begin{align*}
  &\Pr(o_i = b_i, \bar{o}_{-i} = \bar{b}_{-i} \mid \bar{m} = \bar{a}, \lambda = c)\quad  = \quad
  \prod\nolimits_{j=1}^n \Pr(o_j = b_j \mid m_j = a_j, \lambda = c)\\
  &\quad=\quad\Pr(o_i = b_i \mid m_i = a_i, \lambda = c) * \prod\nolimits_{j \neq i}\Pr(o_j = b_j \mid m_j = a_j, \lambda = c) \\
  &\quad=\quad\Pr(o_i = b_i \mid m_i = a_i, \lambda = c) * \Pr(\bar{o}_{-i} = \bar{b}_{-i} \mid \bar{m}_{-i} = \bar{a}_{-i}, \lambda = c) \\
  &\quad=\quad\Pr(o_i = b_i \mid \bar{m} = \bar{a}, \lambda = c) * \Pr(\bar{o}_{-i} = \bar{b}_{-i} \mid \bar{m} = \bar{a}, \lambda = c),
\end{align*}
where the last equality follows by Parameter Independence.
\end{proof}

\subparagraph*{Unifying Relational and Probabilistic Properties.}
An advantage of our framework is that it gives an elegant logical argument why
properties of probabilistic models carry over to the underlying relational
models. This behaviour is, by \cref{theorem:compareSemantics}, an immediate
consequence of the fact that the properties in question are defined in both
cases by syntactically the same
formulae from independence logic. This novel and purely logical perspective complements earlier work on
unifying relational and probabilistic properties employing algebraic and
category-theoretical tools \cite{AbramskyBra11}.

\begin{prop}\label{prop:probVSrelProperty}
  Let $\psi \in \FO(\perp)$ capture a property $P$ of both probabilistic and
  relational models. If a probabilistic model satisfies $P$, its underlying
  relational model also satisfies $P$. For $\psi \in \FO(\dep)$, the converse also holds.
\end{prop}

In particular, this holds for every single property defined in \cref{subsec:defProperties}.
\begin{cor}
  If a probabilistic hidden-variable model satisfies any of the elementary properties
  defined above
  (e.g.~$\lambda$-Independence, Locality,  No-Signalling, \dots),
  its underlying relational model also satisfies the corresponding property.
  For the properties which rely only on dependence atoms (Strong-Determinism,
  Weak-Determinism, Single-Valuedness), the converse also holds.
\end{cor}

The proposition applies of course much more generally than just to the handful of 
properties that we listed explicitly.
  Of course, every family of formulae in $\FO(\perp)$ induces properties of
  relational and probabilistic models in a way that satisfies the assumptions
  of \cref{prop:probVSrelProperty}. Since independence logic is a rather powerful logic --
  recall that it can express all NP properties of teams by \cite{Galliani12} -- this induces a very 
  large class of properties with the aforementioned
  relationship between probabilistic and relational models.

\subsection{Connections between Properties via Logical Entailment}\label{subsec:conProperties}

Implications of the form that ``all models with property $P$ also satisfy property
  $Q$'' have been proved for instance  in \cite{Abramsky13} and \cite{BrandenburgerYan08}. In our
framework, such statements are elegantly cast as entailments between 
formulae of independence logic $\FO(\perp)$. Thus, they can be proven directly on the logical level. 
This allows us to make use of the relationship between $\relEnt$ and $\probEnt$ to
develop the results for both cases in parallel.
\begin{theorem}\mbox{}
\begin{enumerate}
\item Single-Valuedness implies $\lambda$-Independence: 
\[ \SingVal^h_n \allEnt \LInd^h_n.\]
\item Weak Determinism implies Outcome Independence: 
\[ \WeakDet^h_n \allEnt \OutInd^h_n.\]
\item Strong Determinism corresponds to the conjunction of Weak Determinism and
  Parameter Independence: 
  \[ \StrongDet^h_n \equiv_{\text{all}} \WeakDet^h_n \land \ParInd^h_n.\]
\end{enumerate}
\end{theorem}
\begin{proof} By \cref{theorem:compareEntailment}, the probabilistic case suffices for (1) and (2). 
Claim (1)  immediately
follows from the fact that point distributions are stochastically independent of all
  distributions. 

  \medskip
For (2) we need to show that
  $\WeakDet^h_n \probEnt \OutInd^h_n$. Let $\XX = (X,\Pr) \models\WeakDet^h_n $. 
  For every $(\bar{a},c) \in X(\bar{m},\lambda)$ the distribution $\Pr(\bar{o} \mid \bar{m} = \bar{a}, \lambda
  = c)$ is a point distribution. This implies that $\XX \models o_i \perp_{\bar{m}
    \lambda} \bar{o}_{-i}$ and thereby $\XX \models \OutInd^h_n$.

  \medskip
For (3) we first observe that the entailment $\StrongDet^h_n \models \WeakDet^h_n$ is obvious in both semantics.
    For $\StrongDet^h_n \allEnt  \WeakDet^h_n \land \ParInd^h_n$, it remains to show 
    that $\StrongDet^h_n \probEnt \ParInd^h_n$. But this is trivial as well since the distribution over $o_i$
    for a fixed pair  $(m_i ,\lambda)$ assigns, by Strong Determinism, probability 1 to a single value.
    This suffices because point distributions are stochastically independent of
    every other distribution and thus in particular from the conditional
    distribution on $\bar{m}_{-i}$.

  \medskip
    For the converse, we show that $\WeakDet^h_n \land \ParInd^h_n \relEnt  \StrongDet^h_n$.
    Let $X \models \WeakDet^h_n \land \ParInd^h_n$, and let $s_1,s_2 \in X$ agree on $m_i$ and
    $\lambda$. Then, by Parameter Independence, there exists $s_3 \in X$ that takes the same value on 
    $m_i$ and $\lambda$
    and also satisfies $s_3(o_i) = s_1(o_i)$ and $s_3(\bar{m}_{-i}) =
    s_2(\bar{m}_{-i})$. This entails that  $s_3(\bar{m}) = s_2(\bar{m})$ which in turn
    gives, by Weak Determinism, that $s_3(\bar{o}) = s_2(\bar{o})$ and in particular $s_1(o_i) =
    s_3(o_i) = s_2(o_i)$.
\end{proof}

  It is instructive to think about the intuition for the classification of
  Strong Determinism as conjunction of Weak Determinism and Parameter
  Independence: Consider a model which satisfies Weak Determinism but not Strong
  Determinism. In this case, full measurements $\bar{m}$ uniquely determine
  outcomes $\bar{o}$ but single $m_i$ do not determine $o_i$. However, it is then
  necessary that the other measurements $\bar{m}_{-i}$ contain information about
  $o_i$ which contradicts Parameter Independence. In a sense, the independence
  of other measurements is precisely the missing component to go from Weak
  Determinism to Strong Determinism.

  \medskip
Recall that we defined Locality as conjunction of Parameter Independence and
Outcome Independence. Thus, we immediately obtain the following.
\begin{cor}
  Strong Determinism implies Locality: $\StrongDet^h_n \allEnt \Loc^h_n$.
\end{cor}

Because of \cref{prop:probFormulaeEmp}, we can formulate implications
of the form ``if $h$ is a hidden-variable model satisfying $P^h$, then its
  underlying empirical model $e$ satisfies $P^e$'' also by entailment of properties.

\begin{theorem}
The underlying empirical models of hidden-variable models that satisfy
  Parameter Independence and $\lambda$-Independence fulfill No-Signalling:
  \[ \LInd^h_n \land \ParInd^h_n \allEnt \NoSig^e_n.\]
\end{theorem}

Before we prove this we want to give some intuition.
  Recall that the only difference between Parameter Independence and
  No-Signalling is that the former requires independence of $o_i$ and
  $\bar{m}_{-i}$ conditional on $(m_i,\lambda)$ instead of just $m_i$. Intuitively,
  we use the $\lambda$-Independence to obtain a ``copy'' of $s_2$ -- namely
  $\tilde{s}_2$ -- which agrees with $s_1$ on the hidden-variable to allow us to
  apply Parameter Independence.
  
\begin{proof}[Proof (relational case)]
  Let $X \models \ParInd^h_n \land \LInd^h_n$. Let
  $i \in \set{1,\dots,n}$ and $s_1,s_2 \in X$ with $s_1(m_i) = s_2(m_i)$.
  Choose by $\lambda$-Independence an assignment $\tilde{s}_2 \in X$ with
  $\tilde{s}_2(\lambda) = s_1(\lambda)$ and $\tilde{s}_2(\bar{m}) =
  s_2(\bar{m})$. Then $s_1(m_i,\lambda) = \tilde{s}_2(m_i,\lambda)$ and by applying
  Parameter Independence to $s_1$ and $\tilde{s}_2$, we get $s_3 \in X$ with $s_3(o_i) = s_1(o_i)$ and
  $s_3(\bar{m}_{-i}) = \tilde{s}_2(\bar{m}_{-i}) = s_2(\bar{m}_{-i})$. Thus,
  $X \models \NoSig^e_n$.
\end{proof}
\begin{proof}[Proof (probabilistic case)]
  Let $\XX = (X,\Pr) \models \ParInd^h_n \land \LInd^h_n$
  and $i \in \set{1,\dots,n}$. Let $o \in X(o_i), a \in X(m_i)$,
  and $\bar{a}_{-i} \in X(\bar{m}_{-i})$. We put $\Lambda = X(\lambda)$. We
  have
 \begin{align*}
&\Pr(o_i = o \mid \bar{m} = \bar{a}) \quad =\quad \sum_{c \in \Lambda}\Pr(o_i = o,\lambda = c \mid \bar{m} = \bar{a}) \\
&\quad =\quad\sum_{c \in \Lambda}\Pr(o_i = o \mid \lambda = c, \bar{m} = \bar{a})* \Pr(\lambda = c \mid \bar{m} = \bar{a})\\
&\quad =\quad\sum_{c \in \Lambda}\Pr(o_i = o \mid \lambda = c, m_i = a_i) * \Pr(\lambda = c \mid \bar{m} = \bar{a}) \quad\text{(by $\ParInd^h_n$)}\\
&\quad =\quad\sum_{c \in \Lambda}\Pr(o_i = o \mid \lambda = c, m_i = a_i)* \Pr(\lambda = c \mid m_i = a_i)\quad\text{(by $\LInd^h_n$)}\\
&\quad =\quad \;\Pr(o_i = o \mid m_i = a_i)
\end{align*}
    and thus $\XX \models o_i \perp_{m_i} \bar{m}_{-i}$. 
\end{proof}

\section{Existence of Hidden-Variable Models via Satisfiability}

We now investigate the question whether a given empirical
models admits an empirically equivalent hidden-variable model with certain given properties.
Again, our team semantical framework allows us to treat this in an elegant fashion.
Essentially, the question reduces to the satisfiability of formulae of the
form $\ex \lambda \psi$ on extensions
of the given team by a finite set of values 
for the hidden variable $\lambda$,
where $\psi$ encodes the properties of interest. 

\medskip
Notice that a team $X$ of assignments $s\colon\dom(X)\ra A$
can of course also be understood as a team of assignments
$s\colon\dom(X)\ra A\cup\Lambda$ for any set $\Lambda$, and
since elements of $\Lambda$ do not occur in the team,
this does not change any of the \emph{atomic} 
dependence and independence properties of $X$.
However, for an existential formula $\exists\lambda\phi$  
(or a universal one) the universe of values that are available for $\lambda$
does of course matter. 

\begin{definition} For a team $X$ with values in $A$ and a
set $\Lambda$, we write $X+\Lambda$
for the team with the additional supply $\Lambda$ of
values. We say that a formula $\phi\in\FO(\perp)$ is
\emph{satisfiable by a (finite) extension of $X$} if there exists
a (finite) set $\Lambda$ such that $X+\Lambda \models \psi$.
All of this applies as well to probabilistic teams.
\end{definition}  

The following observation, which we formulate for probabilistic teams, connects
existence of suitable hidden-variable models with team semantics.

\begin{prop}
Let $\XX$ be a probabilistic team representing an empirical model $e_{\XX}$. 
For an arbitrary property $P$ of hidden-variable models
captured by a formula $\psi \in \FO(\perp)$, the following two
statements are equivalent:
\begin{enumerate}
\item There is a hidden-variable model $h_{\Pr}$ which satisfies $P$ and is
    empirically equivalent to $e_\XX$.
\item $\exists\lambda\psi$ is satisfiable by a finite extension of $\XX$.
\end{enumerate}
\end{prop}

\begin{proof}
A hidden-variable model $h_{\Pr}$ (where $\lambda$ takes values in $\Lambda$)
which is empirically equivalent to $e_\XX$
corresponds to a team that can be written as a Skolem extension $\YY = (\XX+\Lambda)[\lambda\mapsto\Ff]$
for a suitable function $\Ff\colon X\to \Delta(\Lambda)$. Meanwhile, every Skolem
extension of $\XX + \Lambda$ represents a hidden-variable
model that is equivalent to $e_{\XX}$. The proposition follows since $\XX +
\Lambda \models \ex \lambda \psi$ expresses precisely that a suitable Skolem
extension of $\XX + \Lambda$ satisfies $\psi$.
\end{proof}

We immediately get the following connection between probabilistic and relational
existence results by our general results on the level of team semantics. 

\begin{prop}\label{prop:compareExtension}
  Let $\psi \in \FO(\perp)$ formalize
  a property $P$ of both relational and probabilistic hidden-variable models.
  If a probabilistic empirical model has an empirically equivalent hidden-variable model
  satisfying $P$, then this also holds for its induced relational model.
  If $\psi\in\FO(\dep)$, the converse also holds.
\end{prop}

\begin{proof}
  Assume that the probabilistic case holds for a probabilistic team $\XX =
  (X,\Pr)$ over $\VarE$. Thus, there is a finite set $\Lambda$
  such that $\XX + \Lambda \models \exists\lambda \psi$. This implies by
  \cref{theorem:compareSemantics} that $X + \Lambda \models \exists\lambda
  \psi$, which was to be
  shown. For $\psi\in\FO(\dep)$, the converse follows by similar arguments.
\end{proof}

Again, our elementary properties can be treated as a special case of
this general result.
If a probabilistic empirical model possesses an equivalent hidden-variable
model satisfying any of the elementary properties defined above
(such as Locality, $\lambda$-Independence, Strong Determinism, \dots),
then the same holds for its underlying relational model.
For the properties that rely only on dependence atoms (Strong-Determinism,
Weak-Determinism, Single-Valuedness), the converse also holds.

\medskip
We next give logical proofs of some existence theorems -- which are already known from
\cite{Abramsky13} and \cite{BrandenburgerYan08} -- in our framework.
We will start with a trivial one. 
\begin{prop}
  Every (relational and probabilistic) empirical model
  is realized by an equivalent hidden-variable model satisfying Single-Valuedness.
\end{prop}
\begin{proof}
Recall that $\SingVal^h_n:=\dep(-,\lambda)$,
and for any team $X$, we can take an arbitrary set $\Lambda$
to get that $X+\Lambda\models \exists \lambda\dep(-,\lambda)$.
This establishes the relational case, which also implies the probabilistic one.
\end{proof}

\begin{prop}\label{prop:existenceSD}
Every (relational and probabilistic) empirical model is realized by an
 empirically equivalent hidden-variable model satisfying Strong Determinism.
\end{prop}
\begin{proof}
 Let $X$ be a team and put  $\Lambda := X$. We need to show that
  \[  X+\Lambda \models \ex \lambda \bigwedge\nolimits_{i=1}^n \dep(m_i\lambda,o_i) .\] 
Consider the Skolem extension $Y=(X+\Lambda)[\lambda\mapsto F]$ with
$F(s)= \set{s}$. For all $s,s'\in Y$,
 $s(\lambda) = s'(\lambda)$ implies that $s = s'$. This shows that $Y$ satisfies
Strong Determinism and thus proves the relational case.
Since $\StrongDet^h_n\in\FO(\dep)$, the probabilistic case also follows.
\end{proof}

This construction shows that Strong Determinism is hardly a satisfying property
in itself; letting every fixed hidden variable deterministically determine a
measurement-outcome pair is far from desirable. A natural additional assumption
is $\lambda$-Independence, which states that the measurement process is
independent from the hidden-variables of the system to be measured. This is by
no means satisfied in the construction above; indeed, the hidden-variables
deterministically determine the measurement taken.

\medskip
\Cref{prop:existenceSD} implies a result for Locality
since $\StrongDet^h_n \allEnt \Loc^h_n$.
\begin{cor}
  Every (relational and probabilistic) empirical model is realized by an
  empirically equivalent hidden-variable model satisfying Locality.
\end{cor}

Terms such as ``local realism'' appearing in the literature refer in fact to
  the existence of hidden-variable models satisfying the conjunction of Locality
  and $\lambda$-Independence. This is not proven above and indeed
  in general false as we shall see in \cref{sec:bell} when we discuss
  Bell's Theorem. But if we weaken Strong Determinism to Weak Determinism, we can realize it together with
$\lambda$-Independence. 

\begin{prop}\label{prop:existenceWDlI}
  Every relational empirical model and every probabilistic empirical model with
  rational probabilities is realized by an empirically equivalent hidden-variable model
  satisfying Weak Determinism and $\lambda$-Independence.
\end{prop}
\begin{proof}
 Note that we only need to prove the probabilistic case, as the restriction to
  rational probabilities does not impair the argument in the proof of \cref{prop:compareExtension}.
  We adapt a construction from \cite{BrandenburgerYan08} to our team semantic framework.
    
  \medskip
  Let $\XX = (X,\Pr_{\XX})$ be a probabilistic team over $\VarE$,
  with $| X(\bar{m}) | = L$ and  $| X(\bar{o}) | = K$.
  For every pair $z=(\bar a,\bar b)\in X(\bar m)\times X(\bar o)$, let 
  \[ \Pr_{\XX}(\bar{o} = \bar{b} \mid \bar{m} = \bar{a}) = p(z) =
    \frac{r(z)}{s(z)} \]
  with $r(z), s(z) \in \nats$ such that $s(z)\neq 0$ and $r(z)$ and $s(z)$ are co-prime for each $z$. 
  Let $N$ be the least common multiple of all numbers $s(z)$,
  and choose a set $\Lambda$ with $N$ points. The idea is to make all $\lambda \in \Lambda$ equally likely and
  independent from the measurements. However, we assign a different number of
  hidden-variables to each measurement-output pair as follows. 
  For every $z$, let $N(z):=p(z)N \in \nats$. Clearly, for every $\bar a\in X(\bar m)$ we have that
  $\sum_{\bar b\in X(\bar o)} N(\bar a,\bar b)=N$.
  For every $\bar a\in X(\bar m)$ we thus get a partition of $\Lambda$ into
  a collection $(\Lambda(\bar a,\bar b))_{\bar b\in X(\bar o)}$ of disjoint sets
  where $\Lambda(z)$ has $N(z)$ elements.
  We now define the Skolem extension $\YY=\XX[\lambda\mapsto\Ff]$ by the function
  $\Ff:X\to\Delta(\Lambda)$ that maps $s\in X$ with $s(\bar m,\bar o)=z$ to the uniform
  distribution over $\Lambda(z)$.
 We claim that $\YY$ satisfies Weak Determinism and $\lambda$-Independence.

\medskip\noindent{\bf Weak Determinism:} Weak Determinism follows immediately
from the construction of $\YY$. Regard arbitrary $\bar{a} \in X(\bar{m}), c \in
\Lambda$. There exists a unique $\bar{b} \in X(\bar{o})$ such that $c \in
N((\bar{a},\bar{b}))$. This is by construction the unique $\bar{b} \in X(\bar{o})$ with 
  $\Pr(\bar{o} = \bar{b}, \lambda = c \mid \bar{m} = \bar{a}) > 0$.

\medskip\noindent{\bf $\lambda$-Independence:} 
To prove $\lambda$-Independence we have to show that $\YY\models \bar m\perp
\lambda$. Regard arbitrary $\bar{a} \in X(\bar{m})$ and $c \in \Lambda$ and
choose the unique $\bar{b}$ with non-zero probability implied by Weak
Determinism. We observe that
  \begin{align*}
    &\Pr(\lambda = c \mid \bar{m} = \bar{a}) \ =\ \Pr(\bar o = \bar b , \lambda = c \mid \bar m = \bar a) \\
   &\qquad = \quad \Pr(\lambda = c \mid \bar{o} = \bar{b}, \bar{m} = \bar{a})\cdot \Pr(\bar{o} = \bar{b} \mid \bar{m} = \bar{a})
      = \frac{1}{N(z)} p(z) = \frac{1}{N},
  \end{align*}
independent of $\bar{a}$. This completes the proof.
\end{proof}

Next we provide a normal form for hidden-variable models satisfying Locality and
$\lambda$-Independence, i.e.~``local realism''. This normal form is well-known
from prior literature, cf.~e.g.~\cite{Fine82b, AbramskyBra11, Abramsky13} --
nevertheless we provide an elementary and accessible proof for the benefit of
the reader new to the study of hidden-variables.

\begin{theorem}\label{theorem:localNormalForm}
  Every relational empirical model and every probabilistic empirical model with
  rational probabilities which admits an empirically equivalent hidden-variable satisfying
  Locality and $\lambda$-Independence also admits an equivalent model satisfying
  Strong Determinism and $\lambda$-Independence.
\end{theorem}

\begin{proof}[Proof (relational case)]
Let $X$ be a team over $\VarE$ such that
$X+\Lambda\models  (\ex \lambda\in\Lambda) \LInd^h_n \land \Loc^h_n$. 
Thus, there is a Skolem extension $Y=(X+\Lambda)[\lambda\mapsto F]$
such that $Y\models \LInd^h_n \land \Loc^h_n$
For $\bar{a} \in X(\bar{m}), c \in \Lambda$ we let
  \[ O(\bar a,c) := \set{\bar{b} \in X(\bar{o}): (\bar a,\bar b,c) \in
      Y(\bar{m},\bar{o},\lambda)}.\]
By $\lambda$-Independence, $O(\bar a,c) \neq\emptyset$.
Also, for each $i \leq n$ and $a \in X(m_i)$, we set
  \[ O_i(a,c) := \set{b \in X(o_i): (a,b,c) \in Y(m_i,o_i,\lambda)}. \]
Note that $Y \models \Loc^h_n$ and \cref{lem:locSemanticsRel} 
imply that $O(\bar a,c) = \prod_{i=1}^n O^i(a_i,c)$.
  
\medskip
  Let $F^i_c$ be the set of all functions $f_i\colon X(m_i) \to X(o_i)$ where
  $f_i(a) \in O_i(a,c)$ for all $a \in X(m_i)$.
 The set $\Lambda'$ is now defined to contain all pairs $(c,f) = (c,(f_1,\dots,f_n))$
  where $c \in Y(\lambda)$ and $f_i \in F^i_c$ for  $i \leq n$. 

\medskip
  Let $F'\colon X \to \Pp^+(\Lambda')$ be the function
  with $F'(s)= \set{(c,f) \in \Lambda':  f_i(s(m_i)) = s(o_i) \text{ for all $i$} }$. 
  It is straightforward to see that $F'(s)$ is indeed non-empty. 
  We claim that $Z := X[\lambda\mapsto F']$
  satisfies Strong Determinism and $\lambda$-Independence.
  
\medskip\noindent{\bf Strong Determinism:} 
  Let $s,s' \in Z$ with $s(\lambda) =
  s'(\lambda) = (c,f)$ and $s(m_i) = s'(m_i) = a$. By
  definition of $F'$, it holds that $s(o_i) = f_i(s(m_i)) = f_i(s'(m_i))
  = s'(o_i)$.
  
\medskip\noindent{\bf $\lambda$-Independence:} We show the slightly stronger claim that for every
  $\bar{a} \in X(\bar{m})$ and $(c,f) \in \Lambda'$ there exists some $s \in Z$ with
  $s(\bar{m}) = \bar{a}$ and $s(\lambda) = (c,f)$. By definition of $\Lambda'$,
  it holds that $b_i := f_i(a_i) \in O_i(a_i,c)$ for all $i \leq n$. Thus,
  $\bar{b} \in \prod_{i=1}^n O^i(a_i,c) = O(\bar{a},c)$. Choose $s \in X$ with $s(\bar{m}) = \bar{a},
  s(\bar{o}) = \bar{b}$. Since $(c,f) \in F'(s)$, our claim and thereby $\lambda$-Independence follows.
\end{proof}

\begin{proof}[Proof (probabilistic case)]
Let  now $\XX = (X,\Pr_{\XX})$ be a probabilistic team over $\VarE$ with
$\XX+\Lambda\models  (\ex \lambda\in\Lambda) \LInd^h_n \land \Loc^h_n$,
and let $\YY=(\XX+\Lambda)[\lambda\mapsto \Ff]$ be a Skolem extension
with $\YY\models \LInd^h_n \land \Loc^h_n$
 As in the relational case we set for $\bar{a} \in X(\bar{m}), c \in \Lambda$
  \begin{align*}
  O(\bar a,c) &= \set{\bar{b} \in X(\bar{o}): (\bar{a},\bar{b},c) \in Y(\bar{m},\bar{o},\lambda)} \\
    &= \set{\bar{b} \in X(\bar{o}): \Pr_{\YY}(\bar{o} = \bar{b} \mid \bar{m} = \bar{a},\lambda=c) > 0},
  \end{align*}
 which is non-empty by $\lambda$-Independence.
  Also, we let for all $i \leq n$ and $a \in X(m_i)$
  \[ O_i(a,c) := \set{b \in X(o_i): (a,b,c) \in Y(m_i,o_i,\lambda)} \]
such that, by Locality,  $O(\bar{a},c) = \prod_{i=1}^n O_i(a_i,c)$. 

Let $Z_i=Y(m_i)\times Y(o_i)\times\Lambda$. For every triple $z_i=(a_i,b_i,c)\in Z_i$, let
  \[ \Pr_{\YY}(o_i = b_i \mid m_i=a_i, \lambda = c) = p_i(z_i)=\frac{r_i(z_i)}{s_i(z_i)}.\]
where $r_i(z)$ and $s_i(z)$ are co-prime natural numbers with $s_i(z_i)\neq 0$.
Let $N_i$ be the least common multiple of the numbers $s_i(z_i)$,
let $\Lambda_i$ be a set with $N_i$ elements and  
define 
  \[ \tilde{\Lambda} = \Lambda \times (\Lambda_1 \times \dots \times \Lambda_n).\]
 We then put  $N_i(z_i) = p_i(z_i)N_i$ and construct,
 for every pair $(a_i,c)\in Y(m_i)\times\Lambda$ a partition $(\Lambda_i(a_i,b_i,c))_{b_i\in Y(o_i)}$
of  $\Lambda_i$ with $\abs{\Lambda_i(z_i)}=N_i(z_i)$ for all $z_i\in Z_i$.
This is possible because 
\begin{align*}
\sum_{b_i\in Y(o_i)} N_i(a_i,b_i,c)  &= \sum_{b_i\in Y(o_i)} p_i(a_i,b_i,c) N_i\\ 
&= N_i \sum_{b_i\in Y(o_i)} \Pr_{\YY}(o_i = b_i \mid m_i=a_i, \lambda = c) = N_i. 
\end{align*}
We now define a function $\Ff'\colon X \to \Delta(\tilde{\Lambda})$ as follows.
An assignment $s\in X$ with $s(\bar m,\bar o,\lambda)=(\bar a,\bar b,c)$
defines the tuple $(z_1,\dots,z_n)$ where $z_i=(a_i,b_i,c)\in Z_i$.
The function $\Ff'$ maps $s$ to the probability distribution $\Ff'_s$ with
\[  \Ff'_s(c,(c_1,\dots,c_n)) = \frac{\Pr_{\YY}(\lambda = c \mid \bar{m} = \bar{a}, \bar{o} = \bar{b})}
{N_1(z_1)\dots N_n(z_n)},\]
if  $(c_1,\dots,c_n)\in \Lambda_1(z_1)\times\dots\times\Lambda_n(z_n)$,
and   $\Ff'_s(c,(c_1,\dots,c_n))=0$, otherwise.
It is straightforward to see that this is indeed always a probability distribution.
We then define $\ZZ = \XX[\lambda\mapsto\Ff']$ and claim that $\ZZ$ satisfies Strong
Determinism and $\lambda$-Independence:

\medskip\noindent{\bf Strong Determinism:} Regard an arbitrary $a_i \in X(m_i)$
and $(c,(c_1,\dots,c_n)) \in \tilde{\Lambda}$. 
Suppose that $\Pr_{\ZZ}(o_i = b_i \mid m_i = a_i, \lambda
= (c,(c_1,\dots,c_n))) > 0$. By definition of $\Ff'$ it then follows that 
$c_i \in \Lambda_i(z_i)$ for $z_i = (a_i,b_i,c)$. Since $a_i,c,c_i$ are fixed,
this uniquely determines $b_i$ and shows Strong Determinism.

\medskip\noindent{\bf $\lambda$-Independence:} Here, we abuse notation a bit and 
write $\lambda=(\lambda_1,\lambda_2)$ to refer to the
two components of elements of $\tilde{\Lambda}$. Let $\bar a\in X(\bar m)$ and $ (c,(c_1,\dots,c_n))\in \tilde{\Lambda}$.
There is a unique $\bar b\in X(\bar o)$ with $c_i\in \Lambda(a_i,b_i,c)$ for all $i$. Thus,
\begin{align*}
  &\Pr_{\ZZ}(\lambda = (c,(c_1,\dots,c_n)) \mid \bar{m} = \bar{a})
  = \Pr_{\ZZ}(\lambda =  (c,(c_1,\dots,c_n), \, \bar{o} = \bar{b} \mid \bar{m} = \bar{a}) \\
 &\quad =\quad \Pr_{\ZZ}(\lambda_2 = (c_1,\dots,c_n) \mid \lambda_1 = c,  \,\bar{o} = \bar{b}, \, \bar{m} = \bar{a})\cdot
    \Pr_{\ZZ}(\lambda_1 = c, \, \bar{o} = \bar{b} \mid \bar{m} = \bar{a}).
\end{align*}    

\begin{lemma} $\Pr_{\ZZ}(\bar{o} = \bar{b}, \lambda_1 = c \mid \bar{m} = \bar{a})=
 \Pr_{\YY}(\lambda = c)\prod_{i=1}^n p_i(z_i)$.
\end{lemma}

Indeed, $\Pr_{\ZZ}(\lambda_1 = c, \bar{o} = \bar{b},\bar{m} = \bar{a}) =
\Pr_{\YY}(\lambda = c, \bar{o} = \bar{b}, \bar{m} = \bar{a})$ since marginalizing $\ZZ$ to 
$(\bar{m},\bar{o},\lambda_1)$ gives rise to $\YY$.
Hence
\begin{align*}
& \Pr_{\ZZ}(\bar{o} = \bar{b},\, \lambda_1 = c \mid \bar{m} = \bar{a})\quad =\quad  \Pr_{\YY}(\bar{o} = \bar{b},\, \lambda = c \mid \bar{m} = \bar{a}) \\
&\qquad =\quad  \Pr_{\YY}(\bar{o} = \bar{b} \mid \lambda = c, \, \bar{m} = \bar{a})\cdot
   \Pr_{\YY}(\lambda=c  \mid \bar{m} = \bar{a}), 
\end{align*}
and by Locality and $\lambda$-Independence of $\YY$ this coincides with
\[ \prod_{i=1}^n \Pr_{\YY}(o_i = b_i \mid m_i = a_i, \lambda = c)*\Pr_{\YY}(\lambda = c) 
\quad = \quad \Pr_{\YY}(\lambda = c)\prod_{i=1}^n p_i(z_i).\]
This proves the lemma. Putting it together with the equation above we get 
\begin{align*} 
  &\Pr_{\ZZ}(\lambda = (c,(c_1,\dots,c_n)) \mid \bar{m} = \bar{a}) \\ 
&\qquad =\quad  \prod_{i=1}^n \frac{1}{N_i\cdot p_i(z_i)}\cdot \Pr_{\YY}(\lambda = c)\prod_{i=1}^n p_i(z_i) 
  \quad = \quad \frac{\Pr_{\YY}(\lambda=c)}{N_1\cdots N_n},
\end{align*}
which is independent from $\bar a$.
This proves that $\ZZ\models\LInd^h_n$. 

\end{proof}

\section{Bell's Theorem and Non-Locality}\label{sec:bell}

The famous theorem by Bell, originally formulated in the groundbreaking work
\cite{Bell64}, showed that a certain flavor of local hidden-variable theories
cannot reproduce some predictions by quantum mechanics. This property
corresponds to the conjunction of Strong-Determinism and $\lambda$-Independence
-- which is by \cref{theorem:localNormalForm} closely related to the conjunction of
Locality and $\lambda$-Independence -- in our framework. Commonly, the class of
theories refuted by Bell's Theorem is referred to by the name of ``local
  realism'' (cf.~\cite{ClauserShi78}). 

\medskip
Since Bell's work contains one of the most influential results on
  hidden-variables, we think it is worthwhile to discuss how its assumptions
  compare to our formulation. All important assumptions of Bell,
  which are led to a contradiction, are contained in equation (2) from the
  original paper \cite{Bell64}, namely
  \[ P(\vec{a},\vec{b}) = \int
    \Pr(\lambda)A(\vec{a},\lambda)B(\vec{b},\lambda)\,\mathrm{d} \lambda .\]
  There, $\vec{a},\vec{b}$ range over measurements of two components $A$ and
  $B$, $\lambda$ ranges over a space of hidden-variables, and $A$ and $B$ denote
  functions that provide the outcomes of the respective components
  deterministically depending on their arguments. Thus, $P(\vec{a},\vec{b})$
  denotes the expectation value of the product of the outcomes of $A$ and $B$
  when measurements $\vec{a},\vec{b}$ are chosen.
  
  \medskip
  The assumed properties are implicitly encoded in the dependencies of the
  equation above. Writing $A(\vec{a},\lambda)$ and $B(\vec{b},\lambda)$ entails
  the assumption that the outcome of a component is deterministically given by
  the measurement in the component together with the hidden-variable,
  i.e.~Strong Determinism. Writing $\Pr(\lambda)$ instead of $\Pr(\lambda \mid
  \vec{a},\vec{b})$ assumes that the distribution over $\lambda$ is independent
  from the measurements chosen, i.e.~$\lambda$-Independence. Thus, proving that
  Strong Determinism and $\lambda$-Independence cannot generally be realized
  together constitutes a proof of Bell's No-Go-Theorem.

\medskip
The proof of Bell's Theorem is probabilistic in nature. By
\cref{prop:compareExtension}, a relational analogue implies the probabilistic
formulation. Therefore, we provide here a relational no-go theorem based on the Hardy
Paradox from \cite{Hardy93} (instead of the original argument due to Bell),
following the presentation in \cite{Abramsky13}. We add some further remarks.
The significance of Bell's Theorem goes far beyond
the theoretical no-go result as it allows a way to experimentally test
Non-Locality (cf.~for example \cite{ClauserShi78}). 
In this context it should be noted, that the empirical models and teams appearing in this and
related results \cite{Abramsky13} are not just
arbitrary mathematical constructions, but arise from quantum mechanical
states and measurements. We do not go into details here, but refer to
\cite{AbramskyPulVaa21} where, based on \cite{Abramsky13},
the notion of \emph{quantum realisable teams}  has been spelled out explicitly.

\begin{theorem}
  There exists a relational empirical model 
  without an equivalent hidden-variable model satisfying the conjunction of
  Strong Determinism and $\lambda$-Independence.
\end{theorem}
\begin{proof}
  Let $X$ be an arbitrary team over $\Var_2^e$ which satisfies the following
  assumptions:
  \begin{enumerate}
    \item $X(m_1,m_2) = \set{(a_1,b_1),(a_1,b_2),(a_2,b_1),(a_2,b_2)}$.
    \item $X(o_1,o_2) \sub \set{R,G}^2$ where $R \neq G$
    \item $(a_1b_1,RR) \in X$
    \item $(a_1b_2,RR) \not\in X$
    \item $(a_2b_1,RR) \not\in X$
    \item $(a_2b_2,GG) \not\in X$
  \end{enumerate}
  Of course many such teams exist. Assume that $X$ has an equivalent
  hidden-variable model represented by a team $Y$ with 
  $Y \models \StrongDet^h_2 \land \LInd^h_2$. Note that $Y \models \ParInd^h_2$ 
  since $\StrongDet^h_n \relEnt \ParInd^h_n$. Due to
  (3), there exists $s_1 \in Y$ with $s(m_1m_2) = a_1b_1$ and $s(o_1o_2) = RR$.
  Define $c = s_1(\lambda)$.
  There also exists an $s_2 \in Y$ with $s_2(m_1m_2) = a_1b_2$ and $s(\lambda) =
  c$ due to $\lambda$-Independence and $(a_1,b_2) \in X(\bar{m}) =
  Y(\bar{m})$. Applying Parameter Independence on $s_1$ and $s_2$ gives $s_3 \in
  Y$ with $s_3(\bar{m}) = s_2$ and $s_3(o_1) = s_1(o_1) = R$. Assumption (4)
  gives $s_3(o_2) = G$. 
  
  \medskip
  Analogously to $s_2$, we have $s_4 \in Y$ with $s_4(\bar{m}) = a_2b_2$ and
  $s_4(\bar{m}) = c$. Applying Parameter Independence on $s_2$ and $s_4$ gives $s_5 \in Y$ with
  $s_5(\bar{m}) = a_2b_2$ and $s_5(\bar{o}) = RG$. 
  
  \medskip
  Next, we construct $s_6 \in Y$ with $s_6(\bar{m}) = a_2b_1$ and
  $s_6(\bar{o}) = RG$. However, $s_1(m_2) = b_1 = s_6(m_2)$ but $s_1(o_2) = R \neq
  G = s_6(o_2)$ which contradicts Strong Determinism.
\end{proof}

By \cref{theorem:localNormalForm}, we immediately obtain a no-go result for
local realism; formalizing the notion that quantum mechanics is non-local.

\begin{cor}
  There exist a relational empirical model and a probabilistic empirical model 
  without empirically equivalent hidden-variable models satisfying the conjunction of
  Locality and $\lambda$-Independence.
\end{cor}

We want to conclude this section with a historical note: It shall be remarked
that Bell himself did not regard his theorem as a refutation of
hidden-variables. In fact, Bell was a proponent of Bohmian Mechanics introduced
in \cite{Bohm52a,Bohm52b} – an explicitly non-local hidden-variable interpretation of
quantum mechanics. Bell clearly emphasizes (see \cite{Bell81}, p.~53) that the
essential argument of Bell's Theorem does not depend on determinism or any other
property of hidden-variables but just on the kind of correlations quantum
mechanics predicts. Thus, he said, ``it is a merit of the de Broglie-Bohm version to
  bring this out so explicitly that it cannot be ignored'' (\cite{Bell80}, p.~159).

\section{The Kochen-Specker Theorem and Non-Contextuality}\label{sec:ks}

The celebrated Kochen-Specker Theorem about the contextuality of quantum
mechanics is one of the most important results about hidden-variable models. It is
rather difficult to understand, which is also due to the fact that
it is often formulated in an imprecise way and that the notion of (non-)contextuality is used
inconsistently in the literature. We attempt to shed light on the
Kochen-Specker Theorem by discussing several variants of it, 
including a formulation purely in terms of linear algebra
and two different formulations in terms of logics with team semantics,
highlighting (non)-contextuality as a team-semantical property.  

\subsection{The Linear-Algebraic Version of the Kochen-Specker Theorem}  

We first recall the Kochen-Specker Theorem in essentially its original form.
A \emph{measurement context} on a Hilbert space is a set $X$ of observables
(i.e.~Hermitian operators) that are compatible in the sense that they share 
a common eigenbasis (and thus commute). In quantum mechanical terms, the Kochen-Specker Theorem
states that there is a finite set $M$ of observables
such that it is impossible to assign a unique meaningful value to every
observable in $M$ in a  \emph{non-contextual} way, i.e.~independent of its 
measurement context $X\subseteq M$.

\begin{definition}
Let $M$ be a set of linear operators on a Hilbert space $\calH$,
A valuation $v\colon M \to\reals$  \emph{respects the algebraic
structure of $\calH$} if for any collection $\{A_i: i\in I\}$ of
pairwise commuting operators in $M$ we have that
\begin{itemize}
    \item if $\sum_{i\in I} A_i = A \in M$, then  $v(A)=\sum_{i\in I} v(A_i)$, and
    \item if $A_iA_j=A_jA_i\in M$, then $v(A_iA_j)=v(A_i)v(A_j)$.
\end{itemize}
\end{definition}

Note that the algebraic structure needs to be respected by $v$ only inside
of measurement contexts, i.e. for operators that commute. The product of
non-commuting observables, on the other side, is in general not an observable.
The Kochen-Specker says that, for certain small sets of observables, the
algebraic structure cannot be respected even in this weak sense.

\begin{theorem}[Kochen-Specker]
  For every Hilbert space $\calH$ with $\dim \calH \geq 3$, there exists
  a finite set $M$ of Hermitian operators on $\calH$ such that no $v\colon M \to
  \reals$ with $v \neq 0$ respects the algebraic structure of $\calH$.
\end{theorem}

A modern proof \cite{CabelloEstGar96} for $\dim \calH \geq 4$ relies on a simple combinatorial lemma
about orthonormal bases (ONB) of $\reals^4$.

\begin{lemma}\label{lemma:ksProj}
There exists a set $Z=\{a_1,\dots,a_{18}\} \subseteq \reals^4$
and nine sets $X_1,\dots,X_9$, each consisting of four elements of $Z$
that form an ONB of $\reals^4$, such that no subset $S \subseteq Z$
fulfills $\abs{S \cap X_j} = 1$ for all $j$.
\end{lemma}

This lemma relies on an explicit construction of vectors $a_1,\dots,a_{18}\in\reals^4$,
and a collection of subsets $X_1,\dots,X_9$ that use every
$a_i$ exactly twice. That is, for every $1 \leq i \leq 18$ there exist
precisely two indices $j$ with $a_i \in X_j$. Once this construction is
complete, the lemma follows by a simple parity argument. The
explicit construction is not important for our purposes and can be found in \cite{CabelloEstGar96}.

\medskip
The Kochen-Specker Theorem (for $\dim\calH\geq 4$) can now be established as follows:
Each vector $a\in Z$ induces an orthonormal projection operator $A$.
Let $M$ be the collection of all operators $A_i$, induced by $a_i\in Z$, together with
the identity operator $I$.
For an ONB $\set{a,b,c,d}$ the induced projection operators $A,B,C,D$, together
with $I$, form a measurement context of compatible observables, since they share $\set{a,b,c,d}$ 
as a common eigenbasis. 
Let now $v\colon M\ra \reals$ be a valuation that respects the algebraic structure.
We have to show that $v=0$. Otherwise, assume that $v(A)\neq 0$
for some $A\in M$. Since $v(I)\cdot v(A)=v(I\cdot A)=v(A)$ it follows that
$v(I)=1$. Further, every projection operator $A\in M$ is idempotent,
so $v(A)=v(A)^2$ and hence $v(A)\in\{0,1\}$. 
Finally, if $A,B,C,D$ are projection operators corresponding to an orthonormal basis
it holds that $A + B + C + D = I$ and thereby that $v(A) + v(B) + v(C) + v(D)
  = 1$. Since $v(A),\dots,v(D) \in \set{0,1}$, it follows that precisely one
of the four projections is mapped to 1.
But then the set $S:=\{a_i\in Z:  v(A_i)=1\}$ would 
have the property that $| S\cap X_j |= 1$ for all $j$, contradicting \cref{lemma:ksProj}. 

\medskip
For the proof of dimension three, we refer the interested reader to \cite{Held18,KochenSpe67}.
 
\subsection{Non-Contextuality as a Team Semantical Property}

We now formulate the Kochen-Specker theorem in the language of
teams, focusing on the notion of non-contextuality.  
We provide two alternative formulations each of which has its own
advantages. To formulate this as elegantly as possible, we first define  new
team semantic atoms.

\begin{definition} For $k$-tuples $\bar{x}_1,\bar{x}_2$ and $m$-tuples $\bar{y}_1,\bar{y}_2$ of variables, let
\begin{align*}  X \models \dep((\bar{x}_1,\bar{x}_2),(\bar{y}_1,\bar{y}_2)) \text{ if }&\text{for all }s,t \in
  X \text{ with }s(\bar{x}_1) = t(\bar{x}_2),\\ &\text{ also } s(\bar{y}_1) = t(\bar{y}_2).\end{align*}
Further, we define that $X \models \operatorname{nc}(x_1\dots x_k,y)$ if for all $s,t \in X$ with $t(y)
  \in \set{s(x_1),\dots,s(x_k)}$ it follows that $s(y) = t(y)$. We use this to define the atom 
  $\operatorname{ncc}(x_1\dots x_k)$ for \emph{non-contextual choice}, with semantics given by
\[  \operatorname{ncc}(x_1\dots x_k) \equiv \ex y \pr{\blor_{i=1}^k y = x_i
    \land nc(\bar{x},y)}.\]
\end{definition}

Obviously, this new dependence atom is an extension of the regular one, with
$\dep(\bar{x},\bar{y}) \equiv \dep((\bar{x},\bar{x}),(\bar{y},\bar{y}))$. 
The atom $\operatorname{ncc}(x_1\dots x_k)$ for non-contextual choice expresses that one can choose 
for every assignment $s \in X$ an element $s(y) \in \set{s(x_1),\dots,s(x_k)}$ with the additional constraint that
an element that is selected in an assignment $s$ is also selected in every other
assignment $t$ it which it appears. This explains the name ``non-contextual choice''.

\medskip
Since these all atoms are obviously in NP and downwards-closed, it follows 
by \cite{KontinenVaa09} that they are definable in
dependence logic. Explicit formulae for them are given in \cref{appendixFormulae}.

\begin{definition} A team $X$ over variables $\Var_n^e = \set{m_1,\dots,m_n,o_1,\dots,o_n}$  
(and the empirical model it represents) is \emph{non-contextual} if there exists in 
every component $i\leq n$ a valuation $v_i\colon X(m_i) \to X(o_i)$ that is consistent with the
empirical model in the sense that for measurements $\bar{m} = \bar{a}$ the outcome $\bar{o}
  = v_1(a_1)\dots v_n(a_n)$ is possible. Formally, for every  $n \in \nats$, non-contextuality
can be defined by the formula
 \[ \NonContext^e_n := \ex v_1\dots v_n \pr{\bland_{1 \leq i \leq j \leq n}^n \dep((m_i,m_j),(v_i,v_j)) \land
    \bar{m}\bar{v} \sub \bar{m}\bar{o}}. \]
The generalized dependence atom ensures that if the same measurement appears in components
 $i$ and $j$ then the valuations $v_i,v_j$ agree there.
Notice that $\NonContext^e_n$ uses both (extended) dependence atoms and inclusion atoms, so it
is a formula of independence logic. 
\end{definition}

Thus, $X \not\models \NonContext^e_n$ means that there is no single, non-contextual
assignment of values to all measurements that is consistent with the empirical
model. The term \emph{contextuality} refers to the fact that in this situation
measurement outcomes are inherently dependent on their measurement context,
i.e.~other measurements performed simultaneously. There is no way to assign values to measurements 
independent of each other in a consistent way.

\medskip
With this formula we can formulate an analogue of the Kochen-Specker Theorem.
Let $e_i=(\delta_{i1},\dots,\delta_{i4})\in\{0,1\}^4$ with $\delta_{ij}=1$ if, and only if, $i=j$. 

\begin{theorem}[Logical formulation of the Kochen-Specker Theorem]
There exists a team $X$ over variables $\set{m_1,\dots,m_4}$ such that every
extension $Y$ of $X$ to variables $\set{m_1, \dots, m_4,\allowbreak  o_1, \dots, o_4 }$ which
  satisfies the constraint that $Y(\bar{o}) \sub \set{e_1,\dots,e_4}$
violates non-contextuality, i.e.~$Y \not\models \NonContext^e_n$.
\end{theorem}

The Kochen-Specker Theorem gives us a class of empirical models which are 
necessarily contextual: no non-contextual outcome of measurements is possible,
due to the violation of $\NonContext^e_n$. Note that the constraint
$Y(\bar{o}) \sub \set{e_1,\dots,e_4}$ corresponds to the requirement that
valuations must respect the algebraic structure of the measurement operators.
This implies that in each ONB the valuations actually define a choice
of one of the basis vectors in the ONB.
It is a very remarkable feature of the theorem that the measurement setup alone, 
encoded by $X$, suffices to ensure that all compatible extensions are contextual.

\medskip
Using our non-contextual choice atom we can give a different team semantical
formulation of the Kochen-Specker Theorem that makes this aspect of choice more
explicit. While the version given above is closer to the rest of our framework and structurally more 
similar to the general formulation of the Kochen-Specker Theorem,
the alternative formulation below directly corresponds to the
linear algebraic argument formalized by \cref{lemma:ksProj}
and is more compact and elegant.

\begin{theorem}\label{Thm:ncc}
There is a team $X$ over the variables $\set{m_1,\dots,m_4}$ such that $X
\not\models ncc(\bar{m})$. Furthermore, it is possible to choose $X$ with values
in $\reals^4$ so that for every $s \in X$ the set
$\set{s(m_1),\dots,s(m_4)}$ forms an ONB of $\reals^4$.
\end{theorem}

\begin{proof} By \cref{lemma:ksProj} we can form the team $X$ consisting
of  assignments $s_1,\dots s_9$ with values in $\reals^4$ such that for
each $i\leq 9$, the set $B_i:=\set{s_i(m_1),\dots,\allowbreak s(m_4)}$ is an ONB, consisting of four vectors from
the collection $Z=\{a_1,\dots,a_{18}\}$ constructed in the proof of \cref{lemma:ksProj}.
We claim that $X \not\models ncc(\bar{m})$. Otherwise there exists a choice function
$f$ mapping each of the sets $B_1,\dots,B_9$ to one of its elements such that,
whenever $f(B_i)=a$ and $a\in B_k$, then also $f(B_k)=a$.
But this means that the image of $f$ forms a set $S\subseteq Z$ such that
$|S\cap X_i|=1$ for all $i\leq 9$, contradicting \cref{lemma:ksProj}.
\end{proof}

Thus, \cref{Thm:ncc} is a team semantical formulation of \cref{lemma:ksProj}
which readily implies the Kochen-Specker Theorem. 

\medskip
It is an interesting question, how the treatment of (non-)contextuality via team semantics
can be taken further, to more general scenarios. There is a vast literature on various contextuality properties,
and it is not clear that the setup of teams as it is chosen here is adequate for more general settings. 
The scenarios studied here, sometimes called  Bell scenarios, assume different sets of measurements
each of which can be performed independently by some agent or site. A context is a combined choice of
measurements, one by each agent. We represent this by having for each agent a pair of variables,
one of which evaluates to the choices of measurement, the other to the observed values.

There are more general scenarios, such as the Specker triangle or the Peres-Mermin magic square, where
certain subsets of measurements can be performed simultaneously, but not all of them, and
there are constraints on the (Boolean) values that are observed when a possible set of
measurement is performed.  Together, these constraints result in an impossibility of a global Boolean assignment to
all measurements. Abramsky and Brandenburger \cite{AbramskyBra11} propose an approach where
variables represent the measurements themselves (not the agents)  and the values assigned to these
variables correspond the outcomes. A context is a domain of simultaneously performable measurements
(i.e. variables) and a behaviour in this context is a relational or probabilistic team on this domain.
An empirical model is thus no longer described by a team, but by a set of teams with different domains, one for each context.
Contextuality is then the question whether there exist a team on the domain of all variables,
whose set of restrictions to the contexts coincides with the given empirical model. We further refer to 
\cite{Abramsky14,AbramskyCar19, AbramskyBar21} for more details and for connecting these issues
to applications in other areas, including logic, constraint satisfaction problems, databases, etc.
It is an interesting challenge for future research whether logics of dependence and independence can be
sucessfully applied also this more general scenario.


\begin{thebibliography}{10}

\bibitem{Abramsky13}
S.~Abramsky.
\newblock {Relational Hidden Variables and Non-Locality}.
\newblock {\em Studia Logica}, 101(2):411--452, 2013.
\newblock \href {https://doi.org/10.1007/s11225-013-9477-4}
  {\path{doi:10.1007/s11225-013-9477-4}}.

\bibitem{Abramsky14}
S.~Abramsky.
\newblock Contextual semantics: From quantum mechanics to logic, databases,
  constraints, and complexity, 2014.
\newblock \href {http://arxiv.org/abs/1406.7386} {\path{arXiv:1406.7386}}.

\bibitem{AbramskyBar21}
S.~Abramsky and R.~Barbosa.
\newblock The logic of contextuality.
\newblock In {\em 29th {EACSL} Annual Conference on Computer Science Logic,
  {CSL} 2021}, pages 5:1--5:18, 2021.
\newblock \href {https://doi.org/10.4230/LIPIcs.CSL.2021.5}
  {\path{doi:10.4230/LIPIcs.CSL.2021.5}}.

\bibitem{AbramskyBra11}
S.~Abramsky and A.~Brandenburger.
\newblock {The sheaf-theoretic structure of non-locality and contextuality}.
\newblock {\em New Journal of Physics}, 13, 2011.
\newblock \href {https://doi.org/10.1088/1367-2630/13/11/113036}
  {\path{doi:10.1088/1367-2630/13/11/113036}}.

\bibitem{AbramskyCar19}
S.~Abramsky and G.~Car\'u.
\newblock {Non-locality, contextuality, and valuation algebras: a general
  theory of disagreement}.
\newblock {\em Philosophical Transactions of the Royal Society A}, 377(2157),
  2019.

\bibitem{AbramskyPulVaa21}
S.~Abramsky, J.~Puljujärvi, and J.~Väänänen.
\newblock Team semantics and independence notions in quantum physics, 2021.
\newblock \href {http://arxiv.org/abs/2107.10817} {\path{arXiv:2107.10817}}.

\bibitem{Bell64}
J.~S. Bell.
\newblock {On the Einstein Podolsky Rosen paradox}.
\newblock {\em Physics Physique Fizika}, 1(3):195--200, 1964.
\newblock \href {https://doi.org/10.1103/PhysicsPhysiqueFizika.1.195}
  {\path{doi:10.1103/PhysicsPhysiqueFizika.1.195}}.

\bibitem{Bell66}
J.~S. Bell.
\newblock {On the Problem of Hidden-Variables in Quantum Mechanics}.
\newblock {\em Reviews of Modern Physics}, 1966.
\newblock \href {https://doi.org/10.1103/RevModPhys.38.447}
  {\path{doi:10.1103/RevModPhys.38.447}}.

\bibitem{Bell80}
J.~S. Bell.
\newblock {de Broglie-Bohm, Delayed-Choice, Double-Slit Experiment, and Density
  Matrix}.
\newblock {\em International Journal of Quantum Chemistry}, 18(S14):155--159,
  1980.
\newblock \href {https://doi.org/10.1002/qua.560180819}
  {\path{doi:10.1002/qua.560180819}}.

\bibitem{Bell81}
J.~S. Bell.
\newblock {Bertlmann's Socks and the Nature of Reality.}
\newblock {\em Journal de Physique Colloques}, 42(C2):41--62, 1981.
\newblock \href {https://doi.org/10.1051/jphyscol:1981202}
  {\path{doi:10.1051/jphyscol:1981202}}.

\bibitem{Bohm52a}
D.~Bohm.
\newblock {A Suggested Interpretation of the Quantum Theory in Terms of
  ``Hidden'' Variables. I}.
\newblock {\em Physical Review}, 85(2):166--179, 1952.
\newblock \href {https://doi.org/10.1103/PhysRev.85.166}
  {\path{doi:10.1103/PhysRev.85.166}}.

\bibitem{Bohm52b}
D.~Bohm.
\newblock {A Suggested Interpretation of the Quantum Theory in Terms of
  ``Hidden'' Variables. II}.
\newblock {\em Physical Review}, 85(2):180--193, 1952.
\newblock \href {https://doi.org/10.1103/PhysRev.85.180}
  {\path{doi:10.1103/PhysRev.85.180}}.

\bibitem{BrandenburgerYan08}
A.~Brandenburger and N.~Yanofsky.
\newblock {A classification of hidden-variable properties}.
\newblock {\em Journal of Physics A: Mathematical and Theoretical}, 41(42),
  2008.
\newblock \href {https://doi.org/10.1088/1751-8113/41/42/425302}
  {\path{doi:10.1088/1751-8113/41/42/425302}}.

\bibitem{Bub10}
J.~Bub.
\newblock {Von Neumann's ``No Hidden Variables'' Proof: A Re-Appraisal}.
\newblock {\em Foundations of Physics}, 40(9-10):1333–1340, 2010.
\newblock \href {https://doi.org/10.1007/s10701-010-9480-9}
  {\path{doi:10.1007/s10701-010-9480-9}}.

\bibitem{CabelloEstGar96}
A.~Cabello, J.~Estebaranz, and G.~Garc{\'{i}}a-Alcaine.
\newblock {Bell-Kochen-Specker theorem: A proof with 18 vectors}.
\newblock {\em Physics Letters A}, 212(4):183--187, 1996.
\newblock \href {https://doi.org/10.1016/0375-9601(96)00134-X}
  {\path{doi:10.1016/0375-9601(96)00134-X}}.

\bibitem{ClauserShi78}
J.~Clauser and A.~Shimony.
\newblock {Bell's Theorem. Experimental tests and implications}.
\newblock {\em Reports on Progress in Physics}, 41(12):1881--1927, 1978.
\newblock \href {https://doi.org/10.1088/0034-4885/41/12/002}
  {\path{doi:10.1088/0034-4885/41/12/002}}.

\bibitem{deBroglie27}
L.~de~Broglie.
\newblock {La m{\'{e}}canique ondulatoire et la structure atomique de la
  mati{\`{e}}re et du rayonnement}.
\newblock {\em Journal de Physique et le Radium}, 8(5):225--241, 1927.
\newblock \href {https://doi.org/10.1051/jphysrad:0192700805022500}
  {\path{doi:10.1051/jphysrad:0192700805022500}}.

\bibitem{DurandHanKonMeiVir18a}
A.~Durand, M.~Hannula, J.~Kontinen, A.~Meier, and J.~Virtema.
\newblock {Probabilistic Team Semantics}.
\newblock In F.~Ferrarotti and S~Woltran, editors, {\em Foundations of
  Information and Knowledge Systems. FoIKS 2018}, volume 10833 of {\em Lecture
  Notes in Computer Science}, pages 186--206. Springer International
  Publishing, 2018.
\newblock \href {https://doi.org/10.1007/978-3-319-90050-6_11}
  {\path{doi:10.1007/978-3-319-90050-6_11}}.

\bibitem{EinsteinPodRos35}
A.~Einstein, B.~Podolsky, and N.~Rosen.
\newblock {Can Quantum-Mechanical Description of Physical Reality be Considered
  Complete?}
\newblock {\em Physical Review}, 47(10):777--780, 1935.
\newblock \href {https://doi.org/10.1103/PhysRev.47.777}
  {\path{doi:10.1103/PhysRev.47.777}}.

\bibitem{Fine82b}
A.~Fine.
\newblock Joint distributions, quantum correlations, and commuting observables.
\newblock {\em Journal of Mathematical Physics}, 23(7):1306--1310, 1982.
\newblock \href {https://doi.org/10.1063/1.525514}
  {\path{doi:10.1063/1.525514}}.

\bibitem{Galliani12}
P.~Galliani.
\newblock {Inclusion and exclusion in team semantics --- On some logics of
  imperfect information}.
\newblock {\em Annals of Pure and Applied Logic}, 163(1):68--84, 2012.
\newblock \href {https://doi.org/10.1016/j.apal.2011.08.005}
  {\path{doi:10.1016/j.apal.2011.08.005}}.

\bibitem{GallianiVaa13}
P.~Galliani and J.~V{\"{a}}{\"{a}}n{\"{a}}nen.
\newblock {On Dependence Logic}.
\newblock In A.~Baltag and S.~Smets, editors, {\em Johan van Benthem on Logical
  and Informational Dynamics}, volume~5 of {\em Outstanding Contributions to
  Logic}, pages 101--119. Springer, 2014.
\newblock \href {https://doi.org/10.1007/978-3-319-06025-5_4}
  {\path{doi:10.1007/978-3-319-06025-5_4}}.

\bibitem{GeigerPazPea91}
D.~Geiger, A.~Paz, and J.~Pearl.
\newblock {Axioms and algorithms for inferences involving probabilistic
  independence}.
\newblock {\em Information and Computation}, 91(1):128--141, 1991.
\newblock \href {https://doi.org/10.1016/0890-5401(91)90077-F}
  {\path{doi:10.1016/0890-5401(91)90077-F}}.

\bibitem{GraedelVaa13}
E.~Gr\"adel and J.~V\"a\"an\"anen.
\newblock {Dependence and Independence}.
\newblock {\em Studia Logica}, 101(2):399--410, 2013.
\newblock \href {https://doi.org/10.1007/s11225-013-9479-2}
  {\path{doi:10.1007/s11225-013-9479-2}}.

\bibitem{HannulaHirKonKulVir19}
M.~Hannula, {A}. Hirvonen, J.~Kontinen, V.~Kulikov, and J.~Virtema.
\newblock {Facets of Distribution Identities in Probabilistic Team Semantics}.
\newblock In F.~Calimeri, N.~Leone, and M.~Manna, editors, {\em Logics in
  Artificial Intelligence. JELIA 2019}, volume 11468 of {\em Lecture Notes in
  Computer Science}, pages 304--320. Springer, 2019.
\newblock \href {https://doi.org/10.1007/978-3-030-19570-0_20}
  {\path{doi:10.1007/978-3-030-19570-0_20}}.

\bibitem{HannulaKonBusVir20}
M.~Hannula, J.~Kontinen, J.~Van den Bussche, and J.~Virtema.
\newblock Descriptive complexity of real computation and probabilistic
  independence logic.
\newblock In {\em Proceedings of the 35th Annual ACM/IEEE Symposium on Logic in
  Computer Science}, LICS '20, pages 550--563. {Association for Computing
  Machinery}, 2020.
\newblock \href {https://doi.org/10.1145/3373718.3394773}
  {\path{doi:10.1145/3373718.3394773}}.

\bibitem{Hardy93}
L.~Hardy.
\newblock {Nonlocality for two particles without inequalities for almost all
  entangled states}.
\newblock {\em Physical Review Letters}, 71(11):1665--1668, 1993.
\newblock \href {https://doi.org/10.1103/PhysRevLett.71.1665}
  {\path{doi:10.1103/PhysRevLett.71.1665}}.

\bibitem{Held18}
C.~Held.
\newblock {The Kochen-Specker Theorem}.
\newblock In E.~Zalta, editor, {\em The {Stanford} Encyclopedia of Philosophy
  \emph{(Spring 2018 Edition)}}. Metaphysics Research Lab, Stanford University,
  2018.
\newblock URL:
  \url{https://plato.stanford.edu/archives/spr2018/entries/kochen-specker/}.

\bibitem{Hodges97}
W.~Hodges.
\newblock Compositional semantics for a logic of imperfect information.
\newblock {\em Logic Journal of IGPL}, 5(4):539--563, 1997.
\newblock \href {https://doi.org/10.1093/jigpal/5.4.539}
  {\path{doi:10.1093/jigpal/5.4.539}}.

\bibitem{Jarrett84}
J.~Jarrett.
\newblock On the physical significance of the locality conditions in the bell
  arguments.
\newblock {\em Noûs}, 18(4):569--589, 1984.
\newblock URL: \url{http://www.jstor.org/stable/2214878}.

\bibitem{KochenSpe67}
S.~Kochen and E.~Specker.
\newblock {The Problem of Hidden Variables in Quantum Mechanics}.
\newblock {\em Journal of Mathematics and Mechanics}, 17(1):59--87, 1967.
\newblock URL: \url{https://www.jstor.org/stable/24902153}.

\bibitem{KontinenVaa09}
J.~Kontinen and J.~V\"a\"an\"anen.
\newblock {On Definability in Dependence Logic}.
\newblock {\em Journal of Logic, Language, and Information}, 18(3):317--332,
  2009.
\newblock \href {https://doi.org/10.1007/s10849-009-9082-0}
  {\path{doi:10.1007/s10849-009-9082-0}}.

\bibitem{Studeny89}
M.~Studeny.
\newblock {Multiinformation and the Problem of Characterization of Conditional
  Independence Relations}.
\newblock {\em Problems of Control and Information Theory}, 18(1):3--16, 1989.
\newblock URL:
  \url{http://ftp.utia.cas.cz/pub/staff/studeny/multiinformation-PCIT-89.pdf}.

\bibitem{Vaananen07}
J.~V{\"a}{\"a}n{\"a}nen.
\newblock {\em {Dependence Logic: A New Approach to Independence Friendly
  Logic}}.
\newblock London Mathematical Society Student Texts. Cambridge University
  Press, 2007.
\newblock \href {https://doi.org/10.1017/CBO9780511611193}
  {\path{doi:10.1017/CBO9780511611193}}.

\bibitem{Neumann32}
J.~{von Neumann}.
\newblock {\em {Mathematische Grundlagen der Quantenmechanik}}.
\newblock Springer, 1932.
\newblock \href {https://doi.org/10.1007/978-3-642-61409-5}
  {\path{doi:10.1007/978-3-642-61409-5}}.

\bibitem{WongBut00}
M.~Wong and C.~Butz.
\newblock {On the Implication Problem for Probabilistic Conditional
  Independency}.
\newblock {\em IEEE Transactions on Systems, Man, and Cybernetics - Part A:
  Systems and Humans}, 30(6):785--805, 2000.
\newblock \href {https://doi.org/10.1109/3468.895901}
  {\path{doi:10.1109/3468.895901}}.

\end{thebibliography}

\appendix

\section{Explicit Formulae for Extended Dependence and Non-Contextual Choice}\label{appendixFormulae}
\renewcommand{\thesection}{\Alph{section}}

\begin{prop}
  The following formula $\psi \in \FO(\dep)$ is equivalent to the generalized
  dependence atom $\dep((\bar{x}_1,\bar{x}_2),(\bar{y}_1,\bar{y}_2))$:
  \begin{align*}
    \psi = &\all \bar{z_1}\bar{z_2} \all \bar{w_1}\bar{w_2} \ex u_1 \ex u_2 \ex u_3 \big( \bland_{i=1}^3 \dep(\bar{z_1}\bar{z_2}\bar{w_1}\bar{w_2},u_i) \land \big[ \\
    &\phantom{{}\lor{}}\hspace{0.25cm}(u_1 = u_2 = u_3 \land \bar{z_1}\bar{w_1} \mid \bar{x}_1 \bar{y}_1) \\
    &\lor\hspace{0.25cm}(u_1 = u_2 \neq u_3 \land \bar{z_2}\bar{w_2} \mid \bar{x}_2 \bar{y}_2) \\
    &\lor\hspace{0.25cm}(u_1 \neq u_2 = u_3 \land (\bar{z}_1 \neq \bar{z}_2 \lor \bar{w}_1 = \bar{w}_2))\big] \big)
  \end{align*}
\end{prop}
\begin{proof}
  In team semantics, we can compare equality only within an assignment. However,
  we need to compare values of $x_1y_1$ and $x_2y_2$ across different tuples
  $s,t$. We therefore create copies of $\bar{x}_1\bar{y}_1$ and $\bar{x}_2\bar{y}_2$ and
  regard all possible recombinations: $\bar{z}_1\bar{w}_1$ is to be interpreted as
  copy of $\bar{x}_1\bar{y}_1$ and similarly $\bar{z}_2\bar{w}_2$ as copy of
  $\bar{x}_2\bar{y}_2$. With `recombining' we mean that if $\bar{a} \bar{c}$ appears as
  value of $s(\bar{x}_1\bar{y}_1)$ and $\bar{b} \bar{d}$ as value of $t(\bar{x}_2
  \bar{y}_2)$, we want $\bar{a} \bar{b} \bar{c} \bar{d}$ to appear as value of $\bar{z}_1
  \bar{z}_2 \bar{w}_1 \bar{w}_2$ within one assignment $v$. We can then check
  whether $s(\bar{x}_1) = t(\bar{x}_2) \ra s(\bar{y}_1) = t(\bar{y}_2)$ holds by
  checking whether $v(\bar{z}_1) \neq v(\bar{z}_2) \lor v(\bar{w}_1) = v(\bar{w}_2)$ --
  and this is what is done in the third disjunct.

  \medskip
  The role of the $u_i$ and the first two disjuncts is to ensure that we only
  regard values that are actually taken on by $\bar{x}_1\bar{y}_1$ and
  $\bar{x}_2\bar{y}_2$ in the third disjunct.

  \medskip
  We first extend the team via all possible values of
  $\bar{z}_1,\bar{z}_2,\bar{w}_1,\bar{w}_2$ using the universal quantifier. We than
  add flags $u_1,u_2,u_3$ to every assignment that are functionally dependent on
  $\bar{z}_1\bar{z}_2\bar{w}_1\bar{w}_2$. This ensures that all assignments who
  agree on these copies share the same flags. The disjuncts then ensure that all
  assignments with identical flags must be assigned to the same disjunct.
  Therefore, we partition not really our assignments in the disjunction but
  rather the values of our copies.

  \medskip
  The first disjunct handles all copies where the value of $\bar{z}_1\bar{w}_1$
  does not in fact appear as value of $\bar{x}_1 \bar{y}_1$. Similarly, the second
  disjunct handles all copies where the value of $\bar{z}_2\bar{w}_2$ does not
  appear as value of $\bar{x}_2\bar{y}_2$. The elements that cannot be handled in
  either the first or the second disjunct are precisely those of interest to us
  and must satisfy the $\bar{z}_1 = \bar{z}_2 \ra \bar{w}_1 = \bar{w}_2$ condition
  from the final disjunct. This gives us the semantic of the generalized dependence atom.
\end{proof}

A similar approach works for the atom $\operatorname{nc}(x_1\dots x_k,y)$.
Recall that $X \models \operatorname{nc}(x_1\dots x_k,y)$ if for all $s,t \in X$ with $t(y)
  \in \set{s(x_1),\dots,s(x_k)}$ it follows that $s(y) = t(y)$.

\begin{prop}
  The following $\psi \in \FO(\dep)$ is equivalent to $\operatorname{nc}(x_1\dots x_k,y)$:
  \begin{align*}
    \psi = &\all z_1\dots z_k \all w_1w_2 \ex u_1 \ex u_2 \ex u_3 \big( \bland_{i=1}^3 \dep(\bar{z}\bar{w},u_i) \land \big[ \\
    &\phantom{{}\lor{}}\hspace{0.25cm}(u_1 = u_2 = u_3 \land \bar{z} w_1 \mid \bar{x} y) \\
    &\lor\hspace{0.25cm}(u_1 = u_2 \neq u_3 \land w_2 \mid y) \\
    &\lor\hspace{0.25cm}(u_1 \neq u_2 = u_3 \land (w_2 = w_1 \lor \blor_{i=1}^k w_2 \neq z_i))\big]\big)
  \end{align*}
\end{prop}

  The proof is structurally identical to the one above; $\bar{z} w_1$
  serves as a copy of $\bar{x} y$ (to simulate $s$) and $w_2$ as another
  copy of $y$ (to simulate $t$).

\end{document}